\documentclass[preprint,12pt]{elsarticle}

\usepackage{graphicx}
\usepackage{dhucs}
\usepackage{mathrsfs,amsmath,amssymb,amscd,mathtools,amsthm}
\usepackage{array}
\usepackage{hyperref}
\usepackage{subfig}
\usepackage{adjustbox}
\usepackage{multirow}
\usepackage{booktabs}
\usepackage{color}
\usepackage{setspace}
\usepackage{rotating}
\usepackage{soul}
\usepackage{float}
\usepackage{caption}
\usepackage{bbm}
\usepackage{upgreek}
\usepackage[normalem]{ulem}
\usepackage{nicematrix}
\usepackage{nccmath}
\NiceMatrixOptions{code-for-first-col=\scriptstyle,code-for-first-row=\scriptstyle}

\newtheorem{proposition}{Proposition}

\newtheorem{remark}{Remark}

\graphicspath{{./Figures/}}

\journal{Reliability Engineering \& System Safety}

\begin{document}

\begin{frontmatter}

%% Title, authors and addresses

%% use the tnoteref command within \title for footnotes;
%% use the tnotetext command for theassociated footnote;
%% use the fnref command within \author or \affiliation for footnotes;
%% use the fntext command for theassociated footnote;
%% use the corref command within \author for corresponding author footnotes;
%% use the cortext command for theassociated footnote;
%% use the ead command for the email address,
%% and the form \ead[url] for the home page:
%% \title{Title\tnoteref{label1}}
%% \tnotetext[label1]{}
%% \author{Name\corref{cor1}\fnref{label2}}
%% \ead{email address}
%% \ead[url]{home page}
%% \fntext[label2]{}
%% \cortext[cor1]{}
%% \affiliation{organization={},
%%             addressline={},
%%             city={},
%%             postcode={},
%%             state={},
%%             country={}}
%% \fntext[label3]{}

\title{Lifetime Analysis of Circular \textit{k}-out-of-\textit{n}: G Balanced Systems in a Shock Environment}

\author[label1]{Seung Min Baik}
\ead{seung-min.baik@hec.ca}
\affiliation[label1]{organization={Department of Decision Sciences, HEC Montréal},
            addressline={3000 Chemin de la Côte-Sainte-Catherine},
            city={Montréal},
            postcode={H3T 2A7},
            state={Quebec},
            country={Canada}}

\author[label2]{Yongkyu Cho\corref{cor1}}
\ead{yongkyu.cho@kgu.ac.kr}
\affiliation[label2]{organization={Department of Industrial and Management Engineering, Kyonggi University},
            addressline={154-42 Gwanggyosan-ro},
            city={Suwon},
            postcode={16227},
            state={Gyeonggi-do},
            country={South Korea}}
            
\cortext[cor1]{Correspondence to Yongkyu Cho}
            
%% Abstract
\begin{abstract}
    This paper examines the lifetime distributions of circular $k$-out-of-$n$: G balanced systems operating in a shock environment, providing a unified framework for both discrete- and continuous-time perspectives. The system remains functioning only if at least $k$ operating units satisfy a predefined balance condition (BC). Building on this concept, we demonstrate that the shock numbers to failure (SNTF) follow a discrete phase-type distribution by modeling the system's stochastic dynamics with a finite Markov chain and applying BC-based state space consolidation. Additionally, we develop a computationally efficient method for directly computing multi-step transition probabilities of the underlying Markov chain. Next, assuming the inter-arrival times between shocks follow a phase-type distribution, we establish that the continuous-time system lifetime, or the time to system failure (TTF), also follows a phase-type distribution with different parameters. Extensive numerical studies illustrate the impact of key parameters—such as the number of units, minimum requirement of the number of operating units, individual unit reliability, choice of balance condition, and inter-shock time distribution—on the SNTF, TTF, and their variability. 
\end{abstract}

%% Keywords

\end{frontmatter}

%% Add \usepackage{lineno} before \begin{document} and uncomment 
%% following line to enable line numbers
%% \linenumbers

%% main text
%%
\section*{Notation}
    %\begin{itemize}%[\labelsep=3em]	
    \begin{description}%[\labelsep=3em]
        \item[$n$] number of units in a system
        \item[$k$] minimum number of operating units for a functioning system
        \item[$r$] reliability of a unit; one-step survival probability of an operational unit
        \item[$\mathbf{X}$] system state (random) tuple defined as $\mathbf{X}=(X_1,...,X_n)$; $X_i$ is a binary state variable of unit $i$ where $X_i=1$ if unit $i$ is operational, $X_i=0$ otherwise.
        \item[$T$] minimum tie-set of a system; $T=\{i_1,i_2,\dots,i_{|T|}\}\subset\{1,2,...,n\}$ where each $i\in T$ corresponds to the index of a unit that comprises the system
        \item[$\mathcal{T}$] collection of the minimum tie-sets of a system; $\mathcal{T}=\{ T_1,T_2,\dots,T_{|\mathcal{T}|} \} $
        \item[$M$] shock numbers to system failure (SNTF), i.e., discrete-time system lifetime 
        \item[$Y$] duration of time between two consecutive shocks, i.e., inter-shock time
        \item[$Z$] duration of time to system failure (TTF), i.e., (continuous-time) system lifetime
        \item[$\mathcal{X}$] state space of the system; $\mathcal{X}=\{\textrm{x}_1,\textrm{x}_2,\dots,\textrm{x}_N\}$ where each $\textrm{x}$ denotes a possible realization of system state vector $\mathbf{X}$ and $N\equiv|\mathcal{X}|=2^n$, e.g., if $n=3$, $\textrm{x}_1=(1,1,1),\textrm{x}_2=(1,1,0),\textrm{x}_3=(1,0,1),\dots,\textrm{x}_8=(0,0,0)$, and hence $N=2^3=8$.
        \item[$\textbf{P}$] one-step transition probability matrix after the system experiences a shock; $\textbf{P} = \left[P_{\textrm{x}_a,\textrm{x}_b}\right]$ where $P_{\textrm{x}_a,\textrm{x}_b}\equiv\mathbb{P}\left\{\textbf{X} \text{ after shock}=\textrm{x}_b|\textbf{X} \text{ before shock} =\textrm{x}_a\right\}$ for $(\textrm{x}_a, \textrm{x}_b) \in \mathcal{X}^2$
        \item[$\mathcal{X}_\mathrm{BC}$] collection of the nonfailed system states that satisfy the subscripted balance condition (BC), e.g., $\mathcal{X}_\textrm{BC1},\mathcal{X}_\textrm{BC2},$ and $\mathcal{X}_\textrm{BC3}$.
        \item[$\mathbf{P}_\mathrm{BC}$]  one-step transition probability matrix between the nonfailed states in $\mathcal{X}_{\textrm{BC}}$. 
        \item[$N_{\mathrm{BC}}$] number of nonfailed system states considering the subscripted BC; $N_{\mathrm{BC}}=|\mathcal{X}_\mathrm{BC}|$. 
       \item[$E_{\mathrm{BC}}$] set of all the failed system states considering the subscripted BC; $E_\mathrm{BC}=\mathcal{X}\backslash \mathcal{X}_\mathrm{BC}$.
        \item[$\bar{\mathcal{X}}_{\mathrm{BC}}$] reduced system state space defined as $\bar{\mathcal{X}}_\mathrm{BC} = \{\bar{\textrm{x}}_1, \bar{\textrm{x}}_2, \dots, \bar{\textrm{x}}_{N_\mathrm{BC}}\}\cup\{\bar{\mathrm{x}}_{N_\textrm{BC}+1}\}$, where the subset $\{\bar{\textrm{x}}_1, \bar{\textrm{x}}_2, \dots, \bar{\textrm{x}}_{N_\mathrm{BC}}\}$ corresponds one-to-one with the set $\mathcal{X}_\mathrm{BC}$ and the state $\bar{\mathrm{x}}_{N_\mathrm{BC}+1}$ corresponds to a unique absorbing state that consolidates all the failed states included in $E_\mathrm{BC}$.
        \item[$\bar{\mathbf{P}}_\mathrm{BC}$] one-step transition probability matrix between all the states in $\bar{\mathcal{X}}_{\mathrm{BC}}$.
        \item[$\mathbf{e}$] column vector of ones with the appropriate size; $\mathbf{e}=[1, \dots,1]^\top$ 
        \item[$\mathbf{0}$] row vector of zeros with the appropriate size; $\mathbf{0}=[0,\dots,0]$
    \end{description}

\section{Introduction} \label{sec:intro}
    We consider a system in which multiple homogeneous units are equispaced in a circular layout, and the system can perform its intended function only when the operating units maintain a certain balance requirement. For convenience, we abbreviate this type of system as a circular $k$-out-of-$n$: G balanced (CknGB) system for the last of this research study. A representative example of a CknGB system is an unmanned aerial vehicle (UAV) or urban air mobility (UAM) aircraft equipped with multiple rotors, as shown on the left side of Fig.~\ref{fig:CknGB_concept}. In such systems, the concept of $k$-out-of-$n$ systems, which is widely used in the field of reliability engineering, can be applied from the perspective that at least $k$ out of $n$ rotors must be operating to generate the minimum lift force required to keep the aircraft staying in the air.
    
    In addition to the circular arrangement of units, another crucial factor that distinguishes the CknGB system from a typical $k$-out-of-$n$ system is the consideration of the balance concept. In systems like UAVs and UAMs, even if more than $k$ units are operational, the system cannot perform its core function of maintaining flight unless those units are physically balanced. In the same vein, another example of a CknGB system is the balanced engine system of a manned space shuttle~\cite{hirata2000}, where engine pairs are symmetrically positioned around the center. Specifically, during the landing phase of such a space shuttle, if one of the two engines in a pair fails, the remaining engine of the pair must be shut down to maintain the physical balance, thereby increasing the probability of a safe landing without loss of life~\cite{hirata2000}. Since the definition of system balance can be physical, logical, or conceptual, depending on the analytical context, no single definitive methodology exists for evaluating the reliability of CknGB systems~\cite{hua2016a, cho2023}.
    
    The first reliability study of $k$-out-of-$n$ systems incorporating both spatially distributed units and balance conditions was conducted by Sarper and Sauer~\cite{Sarper2002}. They introduced a novel reliability model, later referred to as the $k$-out-of-$n$ \textit{pairs}: G balanced system with spatially distributed units by Hua and Elsayed~\cite{hua2016a,hua2016b,hua2018}. This model was inspired by the balanced engine system of manned space shuttles and demonstrated the feasibility of quantifying the system reliability from both discrete- and continuous-time perspectives~\cite{hirata2000}. After a decade or so, Hua and Elsayed~\cite{hua2016a,hua2016b,hua2018} extended the model to UAV systems equipped with circularly arranged rotors, contributing to advancements in reliability estimation~\cite{hua2016a}, degradation analysis~\cite{hua2016b}, and reliability approximation methodologies~\cite{hua2018}.
    
    Motivated by the $k$-out-of-$n$ pairs: G balanced system with spatially distributed units, Endharta~\textit{et al.}~\cite{endharta2018} were the first to formally define the CknGB system. While Hua and Elsayed~\cite{hua2016a,hua2016b,hua2018} consistently considered a \textit{symmetry}-based balance condition for their reliability investigation, Endharta~\textit{et al.}~\cite{endharta2018} introduced a \textit{proportionality}-based balance condition and developed a method for calculating system reliability by enumerating minimum tie-sets (or minimal path sets). Subsequently, Endharta and Ko~\cite{endharta2020} extended prior research study by considering a situation where operating units evenly share the physical load, and they proposed economic design methodologies and optimal maintenance policies for such systems. More recently, Cho~\textit{et al.}~\cite{cho2023} introduced a simple yet more generalized balance condition for CknGB systems based on the \textit{center of gravity}, demonstrating that this condition can enhance system reliability compared to previously considered balance conditions. In summary, a series of reliability studies from Hua and Elsayed~\cite{hua2016a} to Cho~\textit{et al.}~\cite{cho2023} share a common feature: the adoption of a geometry-based balance concept rooted in the spatial information of units.

    Building on the aforementioned research stream on the reliability analysis of balanced systems, this study examines the lifetime distribution of the CknGB system. For a similar purpose, but adopting a different concept of system balance, a body of studies has utilized the finite Markov chain imbedding approach (FMCIA)~\cite{fu1994,cui2010,wu2013} and its variants as a core analytical methodology. As pioneering studies, Cui~\textit{et al.}~\cite{cui2018} and Cui~\textit{et al.}~\cite{cui2019} applied the Markov process imbedding technique to analyze a so-called $k$-out-of-$n$: F balanced system with $m$ sectors and a balanced system where the balance is evaluated based on the conceptual distance between states of specified units, respectively. Subsequently, Hongda~\textit{et al.}~\cite{gao2019} introduced a repairable system concept to the $k$-out-of-$n$: F balanced system with $m$ sectors and derived its reliability and availability. Based on the analytical results, the authors proposed three different maintenance policies and compared the performances. Similarly, Fang and Cui~\cite{fang2020} introduced start-up uncertainty into the $k$-out-of-$n$: F balanced system with $m$ sectors, inspired by port system in which several subsystems must work in cooperation. 

    Meanwhile, the research on multi-state balanced systems has emerged. Early work by Wu~\textit{et al.}~\cite{wu2020} explored multiple failure criteria in multi-state balanced systems, introducing threshold-based failure models that account for system degradation over time. Zhao~\textit{et al.}~\cite{zhao2020} extended the multi-state balanced system model by incorporating a shock environment, employing a two-step FMCIA to derive the distributional indices of the system lifetime. Further advancements in the field have introduced novel balancing mechanisms. Fang and Cui~\cite{fang2021} investigated multi-state competing risks under degradation processes, highlighting the impact of prolonged unbalanced states on system failure. Another innovative approach by Zhao~\textit{et al.}~\cite{zhao2024a} explored $k$-out-of-$n$: F balanced systems with common bus performance sharing, demonstrating that performance redistribution can improve system reliability without forcing down or restarting units. 

    On the other hand, several studies have considered the balanced systems equipped with protective devices or standby units that can enhance system survivability. For instance, Wang~\textit{et al.}~\cite{wang2023} analyzed reliability models for balanced systems equipped with multi-state protective devices. Similarly, Zhao~\textit{et al.}~\cite{zhao2024b} proposed a balanced system that dynamically switches standby units to maintain stability, considering both unit state- and symmetric position-based balance conditions. More recently, Zhao~\textit{et al.}~\cite{zhao2025} introduced a standby unit replacement strategy utilizing a dedicated standby pool, ensuring continued system balance by replacing failed or degraded units.
    
    The analytical framework underlying this research study shares similarities with those in Cui~\textit{et al.}~\cite{cui2019} and Zhao~\textit{et al.}~\cite{zhao2020}. However, what distinguishes this study from previous literature is its adherence to the geometric concept of balance conditions, as introduced in the seminal works of Hua and Elsayed~\cite{hua2016a,hua2016b,hua2018}. This approach emphasizes the physical characteristics of the motivational real-world balanced systems such as UAVs or UAMs, while simultaneously posing challenges in evaluating system reliability. In this regard, we outline the key contributions of this study as follows:

    \begin{itemize}
        \item We derive the probability distribution of system lifetime for circular $k$-out-of-$n$: G balanced (CknGB) systems in a shock environment, from both discrete- and continuous-time perspectives. Our approach consolidates the concepts of geometric balance conditions, minimum tie-sets, finite Markov chain imbedding approach, and phase-type distributions.
        \item Leveraging the systemic properties of the CknGB system, we develop a computationally efficient technique to mitigate the curse of dimensionality, making the proposed approach more scalable for larger systems.
        \item We present a well-described case study that illustrates an end-to-end application of the proposed approach, along with extensive numerical results on various system configurations. These demonstrate the adequacy of the proposed methods while analyzing the impact of system parameters on the system's lifetime.
    \end{itemize}
    
    The remainder of this paper is organized as follows. Section~\ref{sec:sys_desc} describes the target system and outlines the crucial concept of balanced systems discussed in this study. Section~\ref{sec:modeling_analysis_lifetime} provides a detailed explanation of the modeling and analysis framework for system lifetime, including the computationally efficient techniques integrated into the proposed approach. Section~\ref{sec:descriptive_case} presents a descriptive case study for a specific system configuration, while Section~\ref{sec:num_result} offers extensive numerical analyses for the system with varying parameters. Finally, Section~\ref{sec:conclusion} concludes this paper and suggests possible future research directions.

\section{Target System Description} \label{sec:sys_desc}
    Consider a circular $k$-out-of-$n$: G balanced (CknGB) system in which homogeneous units are arranged in a circular layout, as shown in Fig.~\ref{fig:CknGB_concept}, and the system can function properly only if the operating units remain balanced. This system begins in a fully functioning state, with no failed unit, but over time, it may become unable to perform its function due to failures of units triggered by external random shocks. 
    \begin{figure}[H]
        \centering
        \resizebox{0.9\textwidth}{!}{
        \includegraphics{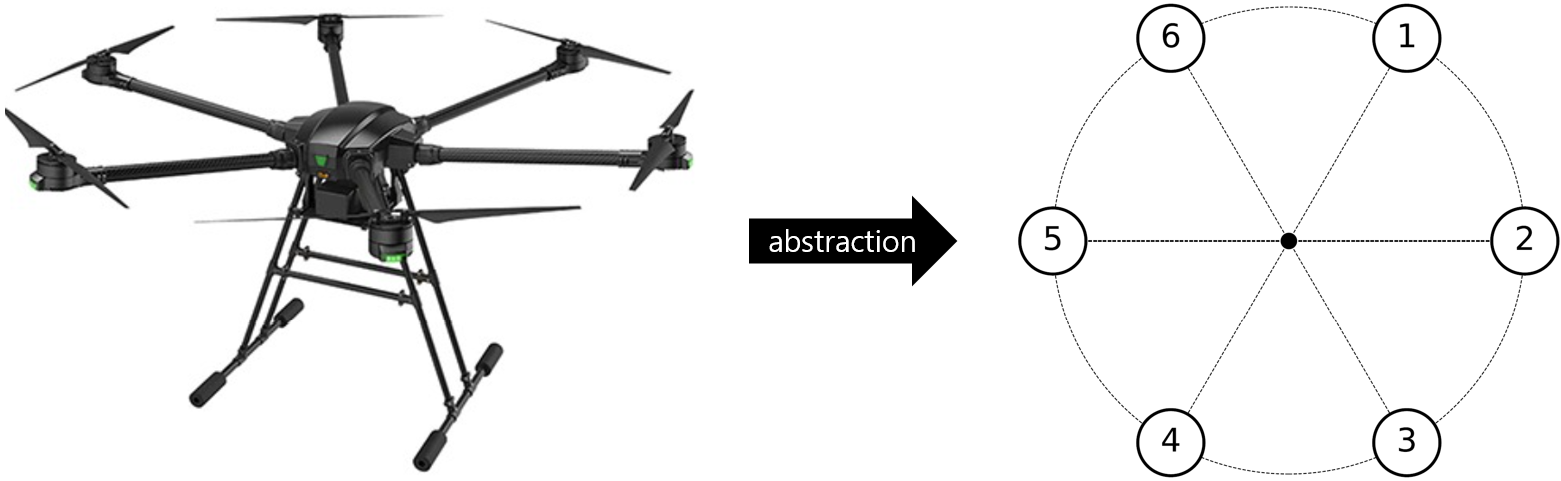}
        }
        \caption{A real-world motivation of the target system and its abstracted graphical model \cite{EFT}}
        \label{fig:CknGB_concept}
    \end{figure}
    
    More specifically, for the system to function, at least $k$ out of the total $n$ units must be operating, and the operating units must also be balanced. Whenever an external shock occurs, each nonfailed unit fails independently of other units with probability $1-r$ and is still operational with probability $r$; $r$ denotes the one-step reliability (i.e., survival probability) of each unit. In this context, from a discrete perspective, the system’s lifetime can be defined as the total number of external shocks the system experiences before it can no longer function properly. From a continuous perspective, it can be defined as the total duration of time until the system experiences that number of shocks. The analytical derivation of the system lifetime distributions for the CknGB system will be explained in detail in Section~\ref{sec:modeling_analysis_lifetime}.

    \subsection{Balance conditions} \label{subsec:bc}
    The balance conditions (BCs) of the CknGB system can be defined in various ways depending on the characteristics of the target system. Here, we provide a brief introduction to the three balance conditions—BC1, BC2, and BC3—summarized by Cho \textit{et al}.~\cite{cho2023}, as illustrated in Fig.~\ref{fig:BCs}. In each illustration in the figure, white circles represent still-operating units, while black circles indicate failed units.
    \begin{figure}[H]
    	\centering
    	\subfloat[][BC1: System is symmetric~\cite{hua2016a}]
    	{
    		\centering\resizebox{0.31\textwidth}{!}{\includegraphics{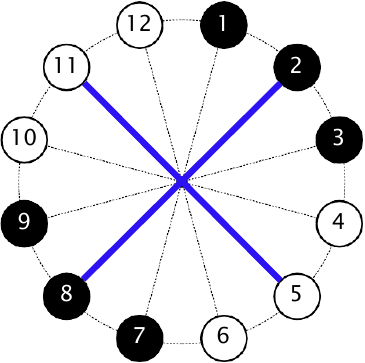}}
    		\label{fig:BC1}		
    	}
    	~
    	\subfloat[][BC2: System is spread proportionately~\cite{endharta2018}]
    	{
    		\centering\resizebox{0.31\textwidth}{!}{\includegraphics{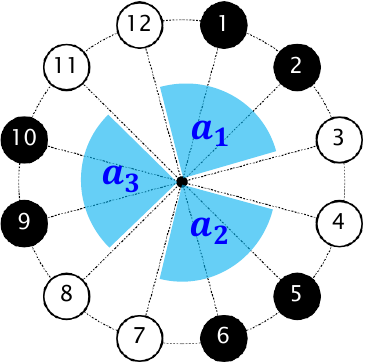}}
    		\label{fig:BC2}		
    	}
    	~
    	\subfloat[][BC3: System has the center of gravity at the origin~\cite{cho2023}]
    	{
    		\centering\resizebox{0.31\textwidth}{!}{\includegraphics{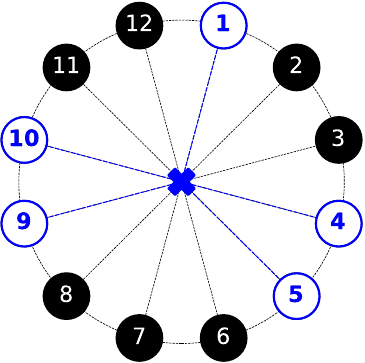}}
    		\label{fig:BC3}		
    	}        
    	\caption{Three different types of balance conditions \cite{cho2023}}
    	\label{fig:BCs}
    \end{figure}
    
    The first balance condition, BC1, is based on \textit{symmetry}, meaning that the system is considered balanced when the operating units form a specific symmetry (e.g., see Fig.~\ref{fig:BC1}). The axis along which the system can be folded in half to align perfectly is called an \textit{axis of symmetry}, and two mutually perpendicular axes of symmetry are referred to as a pair of perpendicular axes of symmetry (shown as the thick blue lines in the figure). The satisfaction of BC1 can be quantitatively evaluated by determining whether at least one pair of perpendicular axes of symmetry exists~\cite{hua2016a}.

    The second balance condition, BC2, is based on \textit{proportionality}, meaning that the system is considered balanced when the operating units are evenly distributed throughout the system (e.g., see Fig.~\ref{fig:BC2}). Whether or not the system satisfies BC2 can be quantitatively verified by examining if the angular gaps (labeled $a_1,a_2$ and $a_3$ in the figure) between the groups of operating units follow a specific pattern~\cite{endharta2018}.

    The third balance condition, BC3, is based on the concept of \textit{center of gravity}, meaning that the system is considered balanced if the center of gravity formed by the operating units is located at the exact center of the system (e.g., see Fig.~\ref{fig:BC3}). Whether or not BC3 is satisfied can be quantitatively examined by calculating the location of the center of gravity (shown by the blue `x' marker in the figure).
    
    \subsection{Rebalancing operation} \label{subsec:ro}
    In the situation where the operation of nonfailed units can be arbitrarily controlled (i.e., an operational unit can be dynamically turned on or off), it is possible to restore a system that has fallen into an unbalanced state due to the failure of some units back to a balanced state. This control action will be referred to as \textit{rebalancing}. 
    
    For example, consider a CknGB system with \( k = 2 \) and \( n = 6 \), where the balance condition BC3 is applied. Fig.~\ref{fig:before_rebalancing} illustrates such a system where units 2 and 4 have failed, leaving the system in an unbalanced state (indicated by the red `x' marker representing the center of gravity), rendering it unable to function. In this situation, if we control unit 6 to stop functioning (depicted in gray to indicate a unit that has not failed but is intentionally turned off), as shown in Fig.~\ref{fig:after_rebalancing}, the center of gravity of the still-operating units (1, 3, and 5) aligns with the origin. Consequently, the system is restored to a balanced state (indicated by the blue `x' marker representing the center of gravity).
    
    Generalizing this principle, whether an unbalanced system can be rebalanced depends quantitatively on whether the set of functioning units includes at least one \textit{minimum tie-set} (also known as \textit{minimal path set}). Therefore, the most critical step in assessing the reliability of a CknGB system is deriving the collection of minimum tie-sets for the system.
    \begin{figure}[H]
    	\centering
    	\subfloat[][Before rebalancing]
    	{
    		\centering\resizebox{0.3\textwidth}{!}{\includegraphics{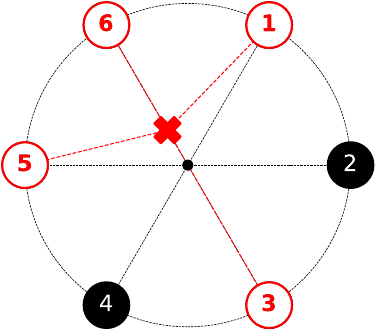}}
    		\label{fig:before_rebalancing}		
    	}
    	~
    	\subfloat[][After rebalancing (unit 6 is forced off)]
    	{
    		\centering\resizebox{0.3\textwidth}{!}{\includegraphics{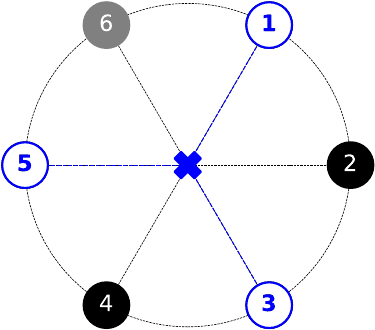}}
    		\label{fig:after_rebalancing}		
    	}
    	\caption{Concept of the rebalancing by turning on/off the operational units}
    	\label{fig:rebalancing}
    \end{figure}

    \subsection{System reliability evaluation using minimum tie-sets} \label{subsec:sys_rel}
    As mentioned in Section~\ref{subsec:ro}, the system reliability of the CknGB system, say $R_S$, can be evaluated using the concept of minimum tie-set. To explain, let $\mathbf{X}=(X_1,X_2,...,X_n)$ be a system state (random) tuple for a system comprised of $n$ independent and identical units where $X_i$ is a binary state variable of unit $i$; $X_i=1$ if unit $i$ is nonfailed and $X_i=0$ otherwise. Then, the system structure function $\phi(\cdot)$ of $\mathbf{X}$ can be defined as follows:
	\begin{align}
		\phi(\mathbf{X}) \equiv 1-\prod_{T \in \mathcal{T}}\left( 1-\prod_{i\in T} X_{i} \right), \label{eq:system_structure_function}
	\end{align}
    where the set $T=\{i_1,i_2,...,i_{|T|}\}$ is a minimum tie-set of the system and the set $\mathcal{T}$ is the collection of all $T$'s. By the definition in Eq.~\eqref{eq:system_structure_function}, $\phi(\mathbf{X})$ is an indicator random variable that equals 1 if $\mathbf{X}$ corresponds to a nonfailed system and 0 otherwise. Therefore, the system reliability $R_S\equiv\mathbb{P}\left\{\phi(\mathbf{X})=1\right\}$ can be directly obtained by calculating $\mathbb{E}[\phi(\mathbf{X})]$ under the assumption $\mathbb{E}[X_i]=r$ for all $i$. Specifically, we have 
    \begin{align}
        R_S=1-\prod_{T \in \mathcal{T}}\left( 1-r^{|T|} \right).
    \end{align}
    Once the set of minimum tie-sets $\mathcal{T}$ is obtained by enumeration, calculating the system reliability is straightforward. For a detailed method of enumerating all the minimum tie-sets, we refer readers to the previous research studies~\cite{endharta2018,cho2023}.
    
    Throughout this section, we have outlined the key factors defining the target system. Before proceeding to the next section, we summarize the main assumptions that apply to the final stage of this research study.
    \begin{itemize}
        \item We consider circular $k$-out-of-$n$: G balanced systems, comprised of $n$ independent and identical units, each having binary states: $1$ for nonfailed (hence operational) state and $0$ for failed state. Therefore, the number of possible system states is $2^n$.
        \item The system remains nonfailed only if there are $k$ or more operating units that satisfy a predefined balance condition. Additionally, each nonfailed unit is subject to on/off control to maintain system balance. Both failed units and turned-off units have no influence on the system's balance. Consequently, any system state that includes at least one minimum tie-set is considered a nonfailed system state.   
        \item The system is subject to sequential random shocks in which each shock probabilistically affects all the units in the system. That is, a shock can be either fatal or nonfatal for each unit independently; after experiencing a shock, each nonfailed unit is still operational with probability $r$ but fails with probability $\tilde{r}\equiv 1-r$. 
    \end{itemize}

\section{Modeling and Analysis of System Lifetime} \label{sec:modeling_analysis_lifetime}
    In this section, we investigate the analytical framework for deriving the lifetime distribution of a circular $k$-out-of-$n$: G balanced system in a shock environment, transitioning from a discrete-time to a continuous-time perspective. As discussed in Section~\ref{sec:sys_desc}, external shocks probabilistically affect the operation of each unit and, consequently, the entire system. Therefore, the system lifetime is a random variable characterized by key parameters, including the number of units $n$, the minimum number of operating units $k$, the one-step survival probability of a unit $r$, and the balance condition BC. 

    To model the system's stochastic dynamics, we construct a discrete-time Markov chain that tracks changes in the system state as it is affected by external shocks. From the discrete-time perspective, we derive the analytically formulated probability mass function (pmf) and the moments of the system lifetime, represented by the discrete random variable $M$, which is referred to as the \textit{shock numbers to system failure} (SNTF). Subsequently, assuming that the inter-arrival times between shocks follow a (continuous) phase-type distribution, which can approximate any positive-valued distribution, we further derive the probability density function (pdf) and the moments of the continuous-time system lifetime. This lifetime is represented by the continuous random variable $Z$, referred to as the \textit{time to system failure} (TTF).

    Overall, our approach closely resembles the well-known finite Markov chain imbedding approach (FMCIA) \cite{fu1994,cui2010,wu2013}. To proceed continuous-time analysis, we adopt its two-step extension, similar to the framework described in Zhao~\textit{et al.}~\cite{zhao2020}. However, in our case, both the system and its individual units are modeled as being either operational (i.e., nonfailed) or failed, resulting in a two-step FMCIA with only two contextual states for both the unit and system. Despite this simple configuration, computational challenges persist due to the large system state space. To address these challenges, we propose several enhancements.

    First, we introduce a consolidation technique to compactify the underlying state space of the Markov chain while preserving analytical accuracy. This consolidation benefits all subsequent analyses, including the derivation of the pmf for SNTF, the pdf for TTF, and their respective moments. Second, we develop alternative methods for calculating multi-step transition probabilities, which further reduce the computational burden specifically for obtaining the pmf of SNTF as well as its moments. Additionally, our methodology is independent of specific balance conditions and can be applied to all balance conditions—BC1, BC2, and BC3—introduced in Section~\ref{subsec:bc}. To reflect this generality, we omit subscripts (e.g., 1, 2, or 3) for balance condition throughout this section, as our analytical framework applies to any balance condition---not only those considered in this paper, but also new conditions that may become of interest in the future, provided that the geometric arrangement of units is relevant.
    
    \subsection*{Step 1: Modeling a Markov chain considering full state space}
    Starting from a fully functioning initial state, the system transitions gradually to one of the failed states as it experiences successive shocks. Each unit is affected by shocks independently, and once a unit fails, it remains in a failed state thereafter. Otherwise, only the most recent shock matters its state change, making each unit’s stochastic behavior Markovian. Consequently, the system’s overall state transitions also exhibit Markovian behavior, making this process naturally suited for modeling as a Markov chain.

    Let us denote the system state by a (random) tuple $\mathbf{X}$, which belongs to the system state space $\mathcal{X} = \{\textrm{x}_1,\textrm{x}_2,\dots,\textrm{x}_N\}$. Each $\textrm{x}$ represents a possible realization of the system state tuple $\mathbf{X}$. For example, when $n=3$, the system state space consists of $\textrm{x}_1=(1,1,1)$, $\textrm{x}_2=(1,1,0)$, ... , and $\textrm{x}_8=(0,0,0)$. These enumerate $N\equiv|\mathcal{X}|=2^3$ possible combinations, ranging from all units functioning to none of them functioning. 
    
    We consider a Markov chain $\{\mathbf{X}_\nu\}_{\nu=0}^{\infty}$ where the $\mathbf{X}_\nu$ denotes the system state after experiencing $\nu$ external shocks. The initial state is given by $\mathbf{X}_0 = \textrm{x}_1 = (1, 1, \dots, 1)$. The one-step transition behavior among the system state space $\mathcal{X}$ can be characterized by the transition probability matrix $\textbf{P} = \left[P_{\textrm{x}_a,\textrm{x}_b}\right]$ where $P_{\textrm{x}_a,\textrm{x}_b}\equiv\mathbb{P}\left\{\textbf{X} \text{ after shock}=\textrm{x}_b|\textbf{X} \text{ before shock} =\textrm{x}_a\right\}$ for $(\textrm{x}_a, \textrm{x}_b) \in \mathcal{X}^2$. 
    
    Note that the size of the state space grows exponentially with the number of units, such that $|\mathcal{X}| = 2^n$. The transition matrix defined over all possible state combinations $\mathcal{X}^2$ consists of $2^{2n}$ elements. For even moderate values of $n$, the size of the matrix might become numerically challenging to handle (e.g., $2^{2n} \approx 10^{7.22}$ for $n=12$). Moreover, the matrix lacks convenient properties, such as symmetry, that could simplify its representation or computation. Therefore, to address these challenges, we will propose state space consolidation approach to mitigate the numerical and computational issues. 
    
    \subsection*{Step 2: Conversion to a smaller Markov chain considering a consolidated state space}
    Let $\mathcal{X}_\textrm{BC}$ denote the collection of nonfailed system states within the full system state space $\mathcal{X}$, where an arbitrary balance condition, denoted by BC, is applied. Accordingly, let $N_\textrm{BC}\equiv|\mathcal{X}_\textrm{BC}|$ represent the number of such states. Then, instead of dealing with the original Markov chain defined on the full state space $\mathcal{X}$, we consider a smaller Markov chain defined on a consolidated state space $\bar{\mathcal{X}}_{\textrm{BC}} = \{\bar{\textrm{x}}_1, \bar{\textrm{x}}_2, \dots, \bar{\textrm{x}}_{N_\mathrm{BC}}\}\cup\{\bar{\mathrm{x}}_{N_\textrm{BC}+1}\}$. Here, the subset $\{\bar{\textrm{x}}_1, \bar{\textrm{x}}_2, \dots, \bar{\textrm{x}}_{N_\mathrm{BC}}\}$ corresponds one-to-one with the set $\mathcal{X}_\mathrm{BC}$, which is merely a reordering of indices from $\mathcal{X}_\mathrm{BC}$ just for ease of representation in the consolidated manner. The distinguished state $\bar{\mathrm{x}}_{N_\mathrm{BC}+1}$ corresponds to a unique absorbing state that consolidates all the failed states included in $E_\mathrm{BC}\equiv\mathcal{X}\backslash \mathcal{X}_\mathrm{BC}$. Fig.~\ref{fig:state_space_reduction_effect} illustrates that the magnitude of $\bar{\mathcal{X}}_{\textrm{BC}}$ can be significantly reduced compared to $\mathcal{X}$ by this consolidation.  
    \begin{figure}[H]
        \centering
        \resizebox{\textwidth}{!}{
        \includegraphics{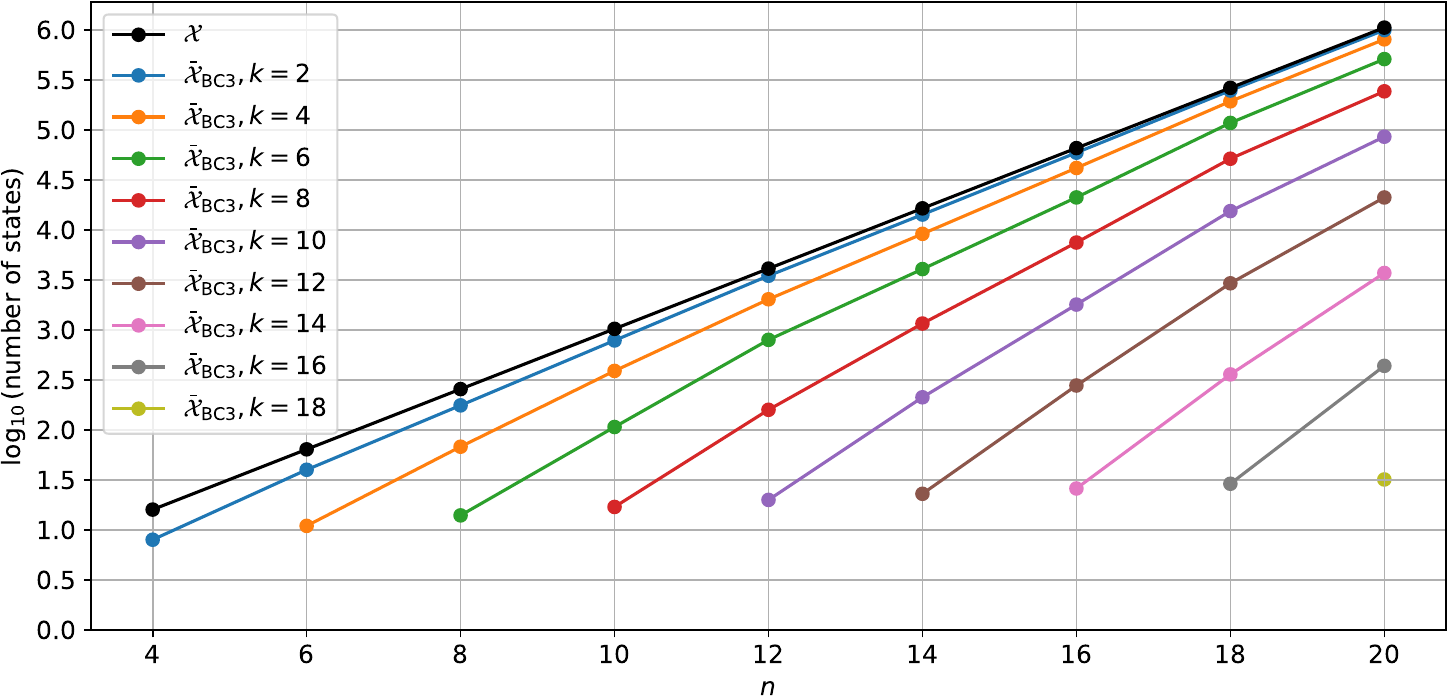}
        }
        \caption{Comparison between the sizes of the full state space $\mathcal{X}$ and the consolidated state spaces $\bar{\mathcal{X}}_\mathrm{BC}$'s for different $k$ and $n$ values where BC3 is applied as a balance condition}
        \label{fig:state_space_reduction_effect}
    \end{figure}

    The transitions of the Markov chain defined on the consolidated system state space $\bar{\mathcal{X}}_{\textrm{BC}}$ can be categorized into three types: 1) the transition between the nonfailed states within $\{\bar{\textrm{x}}_1, \bar{\textrm{x}}_2, \dots, \bar{\textrm{x}}_{N_\mathrm{BC}}\}$ that corresponds to the set $\mathcal{X}_{\textrm{BC}}$, 2) the transition from a nonfailed state (i.e., transient state) within $\mathcal{X}_{\textrm{BC}}$ to the consolidated failed state (i.e. absorbing state) $\bar{\textrm{x}}_{N_\mathrm{BC}+1}$, and 3) the self-transition within the absorbing state $\bar{\textrm{x}}_{N_\mathrm{BC}+1}$. The stochastic behavior of the first type of transition can be represented by a probability matrix $\mathbf{P}_{\textrm{BC}} = [P_{\mathrm{x}_a, \mathrm{x}_b}]$ where $(\mathrm{x}_a, \mathrm{x}_b) \in \mathcal{X}_{\textrm{BC}}^2$. The second type of transition, which occurs from nonfailed states to the absorbing state, can be represented by the probability (column) vector $\mathbf{e} - \mathbf{P}_{\textrm{BC}}\mathbf{e}$, where $\mathbf{e}$ is a column vector of ones with size $|\mathcal{X}_{\textrm{BC}}|$. The third transition is a recursion and thus has probability one. 

    \subsection*{Step 3: Derivation of the SNTF distribution}
    Building upon the above categorization on the system state transitions, we start this step by stating the following remark.
    \begin{remark}
        The transition probability matrix of the Markov chain defined on $\bar{\mathcal{X}}_{\mathrm{BC}}$, denoted by $\bar{\mathbf{P}}_{\mathrm{BC}}$, is expressed as following partitioned matrix:
        \begin{equation}
            \bar{\mathbf{P}}_{\mathrm{BC}} = 
            \begin{bmatrix}
                \mathbf{P}_\mathrm{BC} & \mathbf{e}  - \mathbf{P}_{\mathrm{BC}}\mathbf{e} \\
                \mathbf{0} & 1 
            \end{bmatrix}.
            \nonumber
        \end{equation}        
    \end{remark}

    Note that the SNTF, denoted as $M$, is a \textit{stopping time} of the underlying Markov chain, representing the first passage time to the absorbing state $\bar{\mathbf{x}}_{N_\textrm{BC}+1}$: $M=\inf\left\{ \nu\ge 1 | \mathbf{X}_\nu = \bar{\textrm{x}}_{N_\textrm{BC} + 1} \right\}$. Assuming that we have defined the perfectly functioning state of the consolidated state space to have subscript $1$ (i.e., both $\textrm{x}_1\in\mathcal{X}$ and $\bar{\textrm{x}}_1\in\bar{\mathcal{X}}_\textrm{BC}$ correspond to the system state $(1,1,\dots,1)$), an initial system state probability (row) vector of the consolidated Markov chain becomes $\pmb{\upalpha}_d = [1,0, \dots, 0] \in \mathbb{R}^{|\bar{\mathcal{X}}_{\textrm{BC}}|}$. For the SNTF $M$ to be realized as $m(\geq 1)$, the system state transition should occur among the subset $\{\bar{\textrm{x}}_1, \bar{\textrm{x}}_2, \dots, \bar{\textrm{x}}_{N_\mathrm{BC}}\}$ (or $\mathcal{X}_{\textrm{BC}}$) for $m-1$ times just before transitioning to $\bar{\mathrm{x}}_{N_\textrm{BC}+1}$ (or one of the states within $E_{\textrm{BC}}$) at the $m^\mathrm{th}$ transition. Thus, the pmf of $M$ is given by:
    \begin{equation}
        \mathbb{P}\left\{M=m\right\} = \pmb{\upalpha}_d \mathbf{P}_{\textrm{BC}}^{m-1}(\mathbf{e} - \mathbf{P}_{\textrm{BC}}\mathbf{e}),\quad \forall m \geq 1. \label{eq:pmf}
    \end{equation}
    The above interpretation of the random variable $M$ clearly implies that it follows a discrete phase-type distribution, leading to the following remark.
    \begin{remark}
        A random variable representing the shock numbers to system failure in a circular $k$-out-of-$n$: G balanced system, denoted by $M$, follows a discrete phase-type distribution. That is, $M\sim PH_d(\pmb{\upalpha}_d, \mathbf{P}_\mathrm{BC})$, where the parameters consist of the initial probability vector $\pmb{\upalpha}_d$ and the subtransition matrix $\mathbf{P}_\mathrm{BC}$.
    \end{remark}

    Although the formula for obtaining the distributional indices, such as the pmf given by Eq.~\eqref{eq:pmf}, of the discrete phase-type random variable are well-known and easy to calculate, their direct evaluation can sometimes be computationally burdensome. To be specific, as shown in Fig.~\ref{fig:state_space_reduction_effect}, even within the consolidated state space, the size of the matrix $\mathbf{P}_{\textrm{BC}}$ can be large for large $n$. Consequently, computing the matrix power $\mathbf{P}_{\textrm{BC}}^{m-1}$ quickly becomes intractable as $n$ and $m$ increase due to the computational resource limitations. To address this issue, we will propose an alternative calculation procedure which is computationally more efficient.

    Instead of repeatedly multiplying $\mathbf{P}_{\textrm{BC}}$ to obtain $\mathbf{P}_{\textrm{BC}}^{m-1}$, we illustrate an alternative approach that computes each element of $\mathbf{P}_{\textrm{BC}}^{m-1}$ separately. Let $P^{(m)}_{\textrm{x}_a, \textrm{x}_b}$ be an element of a matrix $\mathbf{P}_{\textrm{BC}}^{m}$, which denotes the $m$-step transition probability from a state $\textrm{x}_a = (x_{a,1}, x_{a,2}, \dots, x_{a,n})$ to a state $\textrm{x}_b = (x_{b,1}, x_{b,2}, \dots, x_{b,n})$, where each $x$ denotes the binary state of unit $i$ for $i\in\{1,2,\dots,n\}$. Since each of the $n$ units is affected by external shocks independently, the $m$-step transition probability can be expressed in a product form as follows: 
    \begin{equation}
        P_{\textrm{x}_a, \textrm{x}_b}^{(m)} = \prod_{i=1}^{n} \mathbb{P}\left\{\text{unit $i$'s state change from $x_{a,i}$ to $x_{b,i}$ after $m$ shocks}\right\}.
        \nonumber
    \end{equation}
    When considering a pair of initial and terminal states for a given unit (i.e., $x_{a,i}$ and $x_{b,i}$), it can take only one of four possible combinations: $(1,1), (1,0), (0,1)$, and $(0,0)$. Let us enumerate each case and examine how shocks would affect the unit. First, for a nonfailed unit to remain still nonfailed (i.e., $(x_{a,i}, x_{b,i})=(1,1)$), all $m$ shocks must be nonfatal. Second, for a nonfailed unit to be failed in the terminal state (i.e., $(x_{a,i}, x_{b,i})=(1,0)$), at least one shock out of $m$ shocks must be fatal. Third, it is impossible for a unit that starts in a failed state to become nonfailed (i.e., $(x_{a,i}, x_{b,i})=(0,1)$) at the end. Lastly, a unit that is already in a failed state will obviously remain failed (i.e., $(x_{a,i}, x_{b,i})=(0,0)$).

    Putting all these together, we have:
    \begin{align*}
        \mathbb{P}\left\{
        \begin{array}{c}
            \text{unit $i$'s state changes from} \\
            \text{$x_{a,i}$ to $x_{b,i}$ after $m$ shocks}
        \end{array}
        \right\} =
        \begin{cases}
            r^m & \textrm{if } (x_{a,i}, x_{b,i}) = (1,1),\\
            1 - r^m & \textrm{if } (x_{a,i}, x_{b,i}) = (1,0),\\
            0 & \textrm{if } (x_{a,i}, x_{b,i}) = (0,1),\\
            1 & \textrm{if } (x_{a,i}, x_{b,i}) = (0,0),
        \end{cases}
        \nonumber
    \end{align*}
    for all $i$ such that $i\in\{1,2,\dots,n\}$. Let us denote the above four state pair values as $\mathrm{y}_1 = (1,1)$, $\mathrm{y}_2 = (1,0)$, $\mathrm{y}_3 = (0,1)$, and $\mathrm{y}_4=(0,0$), and their respective counts as $c_l(\textrm{x}_{a},\textrm{x}_{b}) \equiv \sum_{i=1}^n \mathbf{1}_{\{(x_{a,i},x_{b,i})= \mathrm{y}_l\}}$ for $1\leq l \leq 4$, where $\mathbf{1}_{\{\cdot\}}$ is an indicator function. Then, the $m$-step transition probability simplifies into:
    \begin{equation}
        P_{\textrm{x}_a, \textrm{x}_b}^{(m)} = 
        \begin{cases}
            (r^m)^{c_1(\textrm{x}_{a},\textrm{x}_{b})} (1-r^m)^{c_2(\textrm{x}_{a},\textrm{x}_{b})} & \text{ if $c_3(\textrm{x}_a,\textrm{x}_b) \ge 0$,} \\
            0 &\text{ otherwise.}
        \end{cases}
        \label{eq:transition_prob_element_explicit_formula}
    \end{equation}
    It is noteworthy that the counts of $\mathrm{y}_l$'s can be easily computed and even stored in a data structure in advance for all possible pairs of system states $(\textrm{x}_a, \textrm{x}_b)$'s within $\mathcal{X}_\textrm{BC}^2$.

    \begin{proposition}
        Let $M$ be a random variable representing the shock numbers to system failure in a circular $k$-out-of-$n$: G balanced system. Then, its pmf values can be calculated using the following expression:
        \begin{align*}
            \mathbb{P}\left\{M=m\right\} = \sum_{\mathrm{x}_b \in \mathcal{X}_{\mathrm{BC}}} r^{(m-1)c_1(\mathrm{x}_{1},\mathrm{x}_{b})} \left( r^{c_1(\mathrm{x}_{1},\mathrm{x}_{b})} (1-r^m)^{c_2(\mathrm{x}_{1},\mathrm{x}_{b})} - (1-r^{m-1})^{c_2(\mathrm{x}_{1},\mathrm{x}_{b})} \right)
        \end{align*}
    \end{proposition}
    \begin{proof}
        By applying the explicit formula for calculating the $m$-step transition probability, as derived in Eq.~\eqref{eq:transition_prob_element_explicit_formula}, to the pmf expression given in Eq.~\eqref{eq:pmf}, we obtain the following result.
        \begin{align}
            \mathbb{P}\left\{M=m\right\} &= \pmb{\upalpha}_d \mathbf{P}_{\mathrm{BC}}^{m-1}(\mathbf{e} - \mathbf{P}_{\mathrm{BC}}\mathbf{e}) \nonumber \\ 
            &= \pmb{\upalpha}_d \mathbf{P}_{\mathrm{BC}}^{m-1}\mathbf{e} - \pmb{\upalpha}_d \mathbf{P}_{\mathrm{BC}}^{m}\mathbf{e} \nonumber\\ 
            &= [1, 0,\dots, 0] \mathbf{P}_{\mathrm{BC}}^{m-1} [1, 1, \dots, 1]^\top - [1, 0,\dots, 0] \mathbf{P}_{\mathrm{BC}}^{m} [1, 1, \dots, 1]^\top \nonumber\\ 
            &= \sum_{\mathrm{x}_b \in \mathcal{X}_{\mathrm{BC}}} \left( (r^m)^{c_1(\mathrm{x}_{1},\mathrm{x}_{b})}  (1-r^m)^{c_2(\mathrm{x}_{1},\mathrm{x}_{b})} - (r^{m-1})^{c_1(\mathrm{x}_{1},\mathrm{x}_{b})} (1-r^{m-1})^{c_2(\mathrm{x}_{1},\mathrm{x}_{b})} \right) \nonumber \\ 
            &= \sum_{\mathrm{x}_b \in \mathcal{X}_{\mathrm{BC}}} r^{(m-1)c_1(\mathrm{x}_{1},\mathrm{x}_{b})} \left( r^{c_1(\mathrm{x}_{1},\mathrm{x}_{b})} (1-r^m)^{c_2(\mathrm{x}_{1},\mathrm{x}_{b})} - (1-r^{m-1})^{c_2(\mathrm{x}_{1},\mathrm{x}_{b})} \right). \label{eq:efficient_pmf}
        \end{align}
    \end{proof}
    
    In Eq.~\eqref{eq:efficient_pmf}, note that we only require the values of $c_1(\cdot,\cdot)$ and $c_2(\cdot,\cdot)$ with respect to the initial state $\textrm{x}_1$. Therefore, the previously mentioned precomputation needs to be performed only for the state pairs $(\textrm{x}_1, \textrm{x}_b)$'s for all $\textrm{x}_b$ within $\mathcal{X}_{\textrm{BC}}$. Moreover, this computation can be performed efficiently, as $c_1(\textrm{x}_1, \textrm{x}_b) = \sum_{i=1}^n \textrm{x}_{b,i}$ and $c_2(\textrm{x}_1, \textrm{x}_b) = n - c_1(\textrm{x}_1, \textrm{x}_b)$, given that the initial state $\textrm{x}_1$ is a fully functioning state. Ultimately, when evaluating $\mathbb{P}\left\{M=m\right\}$, we can simply substitute these precomputed values for any $m$, allowing this approach to be efficiently applied even larger $n$.

    Next, the definition of the $p^\mathrm{th}$ moment of the SNTF is given by:
    \begin{equation}
        \begin{aligned}
            \mathbb{E}[M^p] &= \sum_{m=1}^{\infty} m^p \mathbb{P}\left\{M=m\right\}. \nonumber
        \end{aligned}
    \end{equation}
    Plugging the expression derived as Eq.~\eqref{eq:efficient_pmf} into $\mathbb{P}\left\{M=m\right\}$, we obtain
    \begin{equation}
        \mathbb{E}[M^p] = \sum_{m=1}^{\infty} m^p \left( \sum_{\textrm{x}_b \in \mathcal{X}_{\textrm{BC}}} r^{(m-1)c_1(\textrm{x}_{1},\textrm{x}_{b})} \left( r^{c_1(\textrm{x}_{1},\textrm{x}_{b})} (1-r^m)^{c_2(\textrm{x}_{1},\textrm{x}_{b})} - (1-r^{m-1})^{c_2(\textrm{x}_{1},\textrm{x}_{b})} \right) \right).
        \nonumber
    \end{equation}

    For $p=1$, the above expression reduces to $\mathbb{E}[M]$, referred to as the \textit{mean shock numbers to system failure} (MSNTF). It is well-known that the mean and factorial moments of a discrete phase-type distribution can also be expressed explicitly as follows~\citep{bladt2017}:
    \begin{align}
        \mathbb{E}[M] &= \pmb{\upalpha}_d \left(\mathbf{I} - \mathbf{P}_{\textrm{BC}}\right)^{-1} \mathbf{e}\quad \textrm{and}\label{eq:msntf} \\ 
        \mathbb{E}[(M)_p] &= p!  \pmb{\upalpha}_d \left(\mathbf{I} - \mathbf{P}_{\textrm{BC}}\right)^{-p} \mathbf{P}_{\textrm{BC}}^{p-1} \mathbf{e}, \nonumber      
    \end{align}
    where $(M)_p \equiv M(M-1)\cdots(M-p+1)$ denotes the $p^\mathrm{th}$ falling factorial of $M$. 
    
    While these formulas are mathematically compact and exact, they involve matrix inversion and high-dimensional matrix multiplications, which can be computationally expensive, particularly when $\mathbf{P}_{\textrm{BC}}$ is a large-sized matrix. In contrast, the summation-based formula we propose eliminates the need for such operations, making it computationally more efficient and practical for systems with a large state space.

    \subsection*{Step 4: Derivation of the TTF distribution}
    As shown in the previous step, the SNTF $M$ follows a discrete phase-type distribution: $M \sim PH_d (\pmb{\upalpha}_d, \mathbf{P}_{\textrm{BC}})$. If we assume that the inter-shock time, represented by another random variable $Y$, follows a phase-type distribution, then the distribution of the TTF, denoted by the random variable $Z$, also follows a phase-type distribution. In this step, we formally derive this result.

    Assume that the inter-arrival times $Y_m$ between the $(m-1)^\mathrm{th}$ and $m^\mathrm{th}$ external shocks are independent and identically distributed (i.i.d.) random variables for $m=1,2,3, \dots$, all following the same distribution as the random variable $Y$, which follows a phase-type distribution characterized by the initial probability (row) vector $\pmb{\upalpha}_c$ and the subgenerator matrix $\mathbf{T}_c$. Also, we assume that the $0^{\mathrm{th}}$ external shock occurs at time $0$. Then, the TTF $Z$ can be interpreted as the sum of inter-arrival times between the shocks up to the one that causes the system failure. This allows us to express the TTF $Z$ as the compound random variable: $Z = \sum_{m=1}^{M} Y_m$.

    Applying the closure property known for the phase-type distributions, we get the following result.
    \begin{proposition} % bladt2017, pp141, Theorem 3.1.27. 
    Consider a circular $k$-out-of-$n$: G balanced system. If its SNTF $M$ follows a discrete phase-type distribution, that is, $M \sim PH_d(\pmb{\upalpha}_d, \mathbf{T}_d)$, and its inter-shock time $Y$ follows a phase-type distribution, that is, $Y \sim PH_c(\pmb{\upalpha}_c, \mathbf{T}_c)$, then its TTF $Z$ follows a phase-type distribution as follows:
    \[
    Z \sim PH_c\left(\pmb{\upalpha}_d \otimes \pmb{\upalpha}_{c}, \mathbf{I} \otimes \mathbf{T}_c + \mathbf{T}_d \otimes (-\mathbf{T}_c \mathbf{e} \pmb{\upalpha}_c)\right)
    \]
    where $Z=\sum_{m=1}^M Y_m$. \label{prop:cp}
    \end{proposition} 
    \begin{proof}
    Since we have already established that the SNTF $M$ follows a discrete phase-type distribution, the theorem follows directly from the closure properties of phase-type distributions~\citep{bladt2017}. Consequently, $Z$ follows a phase-type distribution, with its parameters composed of those of the discrete- and continuous phase-type distributions, $M$ and $Y$, respectively.
    \end{proof}

    As shown in Proposition~\ref{prop:cp}, the TTF $Z$ follows a phase-type distribution with newly defined parameters $\pmb{\upalpha}_Z = \pmb{\upalpha}_d \otimes \pmb{\upalpha}_{c}$ and $\mathbf{T}_Z = \mathbf{I} \otimes \mathbf{T}_c + \mathbf{P}_{\textrm{BC}} \otimes (-\mathbf{T}_c \mathbf{e} \pmb{\upalpha}_c)$. From the properties of the phase-type distributions, the probability density function of $Z$, $f(z)$, is given by:
    \begin{equation}
        f(z) = - \pmb{\upalpha}_Z \exp \left(z \mathbf{T}_Z \right) \mathbf{T}_Z \mathbf{e}, \label{eq:pdf}
    \end{equation}
    where $\mathbf{e}$ is a vector of ones with the length $K|\mathcal{X}_{\textrm{BC}}|$. Furthermore, the $p^\mathrm{th}$ moment of the TTF, $\mathbb{E}[Z^p]$, can be obtained as~\citep{bladt2017}:
    \begin{equation}
        \mathbb{E}[Z^p] = p! \pmb{\upalpha}_Z\left(-\mathbf{T}_Z^{-1}\right)^{p}\mathbf{e}\quad\mathrm{for}\ p=1,2,\dots .
        \nonumber
    \end{equation}
    In particular, the \textit{mean time to system failure} (MTTF) with $p=1$, $\mathbb{E}[Z]$, is given by:
    \begin{equation}
        \mathbb{E}[Z] = -\pmb{\upalpha}_Z \mathbf{T}_Z^{-1} \mathbf{e}. \nonumber
    \end{equation}
    
    This section has detailed the overall framework for deriving the lifetime distributions of the target system. In the next section, we will present a specific case study and provide a comprehensive explanation of how the previously described procedure is applied to the target system.

\section{A Descriptive Case ($k=2,n=4,r=0.7$, BC3, Erlang)} \label{sec:descriptive_case}
    Let us consider a circular 2-out-of-4: G balanced system in which unit reliability is $r=0.7$, BC3 (i.e., center of gravity at the origin) is applied as a balance condition, and the inter-shock times follow i.i.d. $Erlang(\alpha=2,\lambda=2)$ distribution. The system begins in the perfect system state, $(1,1,1,1)$, and experiences sequential shocks that independently cause unit failures, eventually leading to the system failure. Fig.~\ref{fig:example_dtmc} is the transition probability diagram of the target system as a discrete-time Markov chain (DTMC). 
    
    Using the full state notation in which each $\textrm{x}\in\mathcal{X}$ denotes a system state, the DTMC has $|\mathcal{X}|=2^4=16$ different states comprised of seven transient states, say ${\mathcal{X}}_\mathrm{BC3}=\{\textrm{x}_{1}, \textrm{x}_{2}, \textrm{x}_{3}, \textrm{x}_{5}, \textrm{x}_{6}, \textrm{x}_{9}, \textrm{x}_{11}\}$ (where $\textrm{x}_1$ is the initial state), and nine absorbing states, say $E_\textrm{BC3}=\{\textrm{x}_{4},\textrm{x}_{7},\textrm{x}_{8},\textrm{x}_{10},\textrm{x}_{12},\textrm{x}_{13},\textrm{x}_{14},\textrm{x}_{15},\textrm{x}_{16}\}$. Each absorbing state corresponds to the state that does not include any minimum tie-set hence cannot be operational even by the rebalancing mechanism. Under this setting, the shock number to system failure (SNTF) is defined as the number of transitions the DTMC undergoes, starting from the initial state $\textrm{x}_1$ and continuing until it enters one of the absorbing states.

    To compactify the state space so that we can exploit the useful properties of phase-type distribution, the DTMC is converted to another DTMC defined by a smaller state space $\bar{\mathcal{X}}_\mathrm{BC3}=\{\bar{\mathrm{x}}_1,\bar{\mathrm{x}}_2,...,\bar{\mathrm{x}}_7\}\cup\{\bar{\mathrm{x}}_8\}$ hence $|\bar{\mathcal{X}}_\mathrm{BC3}|=8$. Here, the subset $\{\bar{\mathrm{x}}_1,\bar{\mathrm{x}}_2,...,\bar{\mathrm{x}}_7\}\subset\bar{\mathcal{X}}_{\textrm{BC3}}$ corresponds one-to-one with the set of transient states $\{ \textrm{x}_{1}, \textrm{x}_{2}, \textrm{x}_{3}, \textrm{x}_{5}, \textrm{x}_{6}, \textrm{x}_{9}, \textrm{x}_{11} \}=\mathcal{X}_\textrm{BC3}\subset\mathcal{X}$ and the state $\bar{\mathrm{x}}_8\in\bar{\mathcal{X}}_{\textrm{BC3}}$ consolidates the set of absorbing states $E_\textrm{BC3}\subset\mathcal{X}$. In summary, the new DTMC is an absorbing DTMC with the initial state $\bar{\mathrm{x}}_1$ and an absorbing state $\bar{\mathrm{x}}_8$, which still maintains enough information to obtain the probability distribution of SNTF. Table~\ref{table:example_system_states} enumerates every possible system state in the target system and summarizes the state-wise relationship between the two DTMCs using different notations.

    \begin{figure}[H]
        \centering
        \includegraphics[width = 0.8\textwidth]{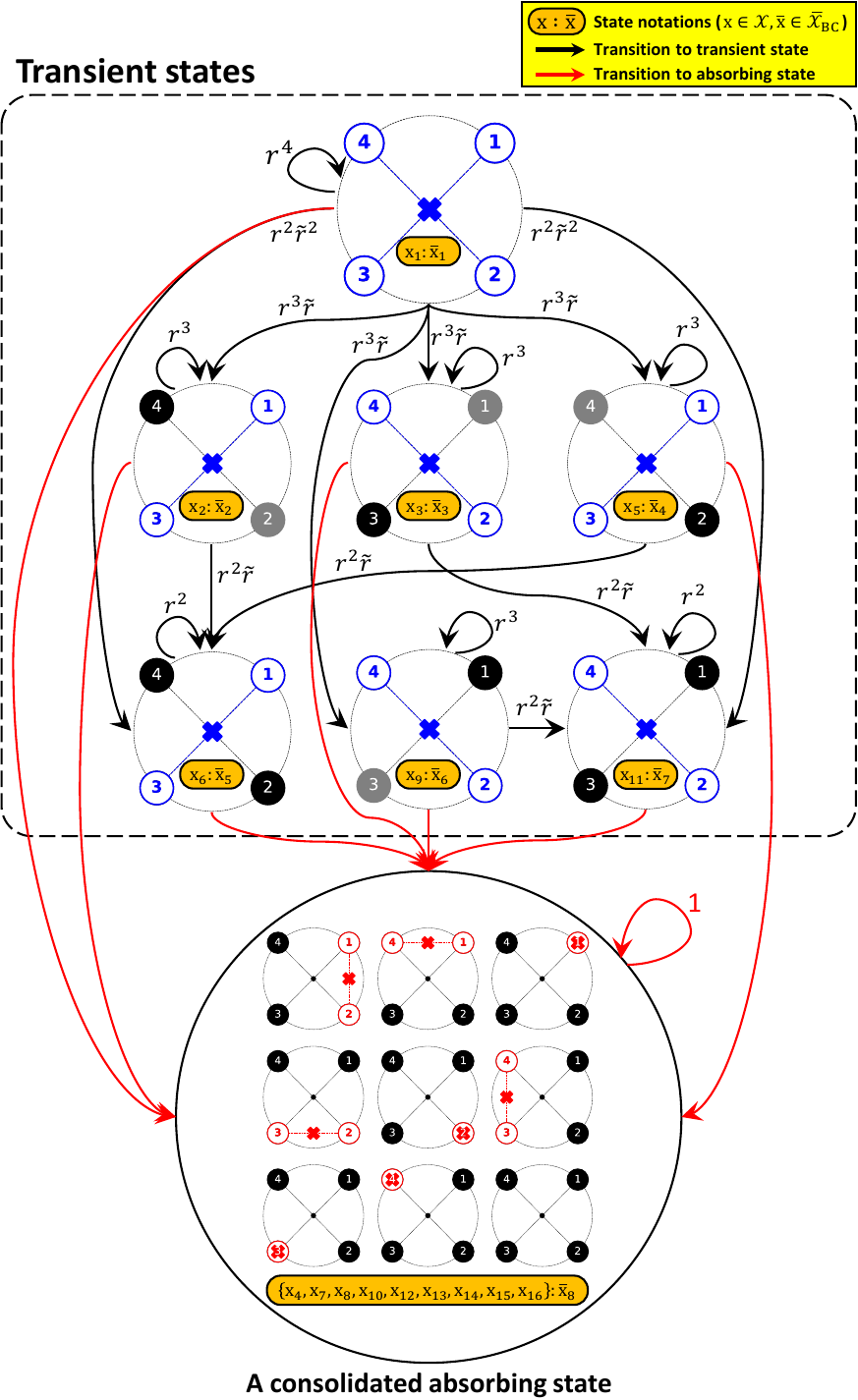}
        \caption{The transition probability diagram for a CknGB system with $k=2$ and $n=4$ where BC3 is applied as a balance condition and $\tilde{r}\equiv1-r$}
        \label{fig:example_dtmc}
    \end{figure}    

    \begin{table}[H]
        \centering
        \caption{System state enumeration $(k=2, n=4, \textrm{BC3})$}
        \resizebox{\textwidth}{!}{%        
        \begin{tabular}{|c|c|c|c|c|}
        \hline
        \textbf{No.} &
          \textbf{System State Tuple} &
          \textbf{\begin{tabular}[c]{@{}c@{}}Full State Notation:\\ $\mathrm{x}\in\mathcal{X}=\{\mathrm{x}_1,\mathrm{x}_2,...,\mathrm{x}_{16}\}$\end{tabular}} &
          \textbf{\begin{tabular}[c]{@{}c@{}}Consolidated State Notation:\\ $\bar{\mathrm{x}}\in\bar{\mathcal{X}}_\mathrm{BC3}=\{\bar{\mathrm{x}}_1,\bar{\mathrm{x}}_2,...,\bar{\mathrm{x}}_7\}\cup\{\bar{\mathrm{x}}_8\}$\end{tabular}} &
          \textbf{System Status} \\ \hline
        1  & $(1,1,1,1)$          & $\mathrm{x}_1$    & $\bar{\mathrm{x}}_1$                           & Nonfailed                                                            \\ \hline
        2  & $(1,1,1,0)$          & $\mathrm{x}_2$    & $\bar{\mathrm{x}}_2$                           & \begin{tabular}[c]{@{}c@{}}Nonfailed\\ (by rebalancing)\end{tabular} \\ \hline
        3  & $(1,1,0,1)$          & $\mathrm{x}_3$    & $\bar{\mathrm{x}}_3$                           & \begin{tabular}[c]{@{}c@{}}Nonfailed\\ (by rebalancing)\end{tabular} \\ \hline
        4  & {$\mathbf{(1,1,0,0)}$} & $\mathrm{x}_4$    & $\mathrm{\mathbf{\bar{x}}}_\mathbf{8}$                  & \textbf{Failed}                                                      \\ \hline
        5  & $(1,0,1,1)$          & $\mathrm{x}_5$    & $\bar{\mathrm{x}}_4$                           & \begin{tabular}[c]{@{}c@{}}Nonfailed\\ (by rebalancing)\end{tabular} \\ \hline
        6  & $(1,0,1,0)$          & $\mathrm{x}_6$    & $\bar{\mathrm{x}}_5$                           & Nonfailed                                                            \\ \hline
        7  & {$\mathbf{(1,0,0,1)}$} & $\mathrm{x}_7$    & \multirow{2}{*}{$\mathrm{\mathbf{\bar{x}}}_\mathbf{8}$} & \textbf{Failed}                                                               \\ \cline{1-3} \cline{5-5} 
        8  & {$\mathbf{(1,0,0,0)}$} & $\mathrm{x}_8$    &                               & \textbf{Failed}                                                      \\ \hline
        9  & $(0,1,1,1)$          & $\mathrm{x}_9$    & $\bar{\mathrm{x}}_6$                           & \begin{tabular}[c]{@{}c@{}}Nonfailed\\ (by rebalancing)\end{tabular} \\ \hline
        10 & {$\mathbf{(0,1,1,0)}$} & $\mathrm{x}_{10}$ & $\mathrm{\mathbf{\bar{x}}}_\mathbf{8}$                  & \textbf{Failed}                                                      \\ \hline
        11 & $(0,1,0,1)$          & $\mathrm{x}_{11}$ & $\bar{\mathrm{x}}_7$                           & Nonfailed                                                            \\ \hline
        12 & {$\mathbf{(0,1,0,0)}$} & $\mathrm{x}_{12}$ & \multirow{5}{*}{$\mathrm{\mathbf{\bar{x}}}_\mathbf{8}$} & \textbf{Failed}                                                      \\ \cline{1-3} \cline{5-5} 
        13 & {$\mathbf{(0,0,1,1)}$} & $\mathrm{x}_{13}$ &                               & \textbf{Failed}                                                      \\ \cline{1-3} \cline{5-5} 
        14 & {$\mathbf{(0,0,1,0)}$} & $\mathrm{x}_{14}$ &                               & \textbf{Failed}                                                      \\ \cline{1-3} \cline{5-5} 
        15 & {$\mathbf{(0,0,0,1)}$} & $\mathrm{x}_{15}$ &                               & \textbf{Failed}                                                      \\ \cline{1-3} \cline{5-5} 
        16 & {$\mathbf{(0,0,0,0)}$} & $\mathrm{x}_{16}$ &                               & \textbf{Failed}                                                      \\ \hline
        \end{tabular}
        }
        \label{table:example_system_states}
    \end{table}

    For the DTMC with the reduced state space $\bar{\mathcal{X}}_\mathrm{BC3}$, its transition probability matrix $\bar{\mathbf{P}}_\textrm{BC3}$ can be explicitly expressed by the following $(8\times 8)$-sized partitioned square matrix where $r$ is the one-step reliability of a unit and $\tilde{r}\equiv 1-r$: 

    \begin{align}
        \bar{\mathbf{P}}_\mathrm{BC3} &=  \begin{bmatrix}
            \mathbf{P}_\mathrm{BC3} & \mathbf{e}-{\mathbf{P}}_\mathrm{BC3}\mathbf{e} \\ 
            \mathbf{0} & 1
            \end{bmatrix} \nonumber\\ 
            &= \begin{bNiceArray}{cccccccl}[first-col,first-row]
                	&\bar{\mathrm{x}}_1	&\bar{\mathrm{x}}_2	&\bar{\mathrm{x}}_3	&\bar{\mathrm{x}}_4	&\bar{\mathrm{x}}_5	&\bar{\mathrm{x}}_6	&\bar{\mathrm{x}}_7	&\bar{\mathrm{x}}_8\\
                \bar{\mathrm{x}}_1	&r^4	&r^3\tilde{r}	&r^3\tilde{r}	&r^3\tilde{r}	&r^2\tilde{r}^2	&r^3\tilde{r}	&r^2\tilde{r}^2	&1-(r^4 + 4r^3\tilde{r} + 2r^2\tilde{r}^2)\\
                \bar{\mathrm{x}}_2	&0	&r^3	&0	&0	&r^2 \tilde{r}	&0	&0	&1-(r^3 + r^2\tilde{r})\\
                \bar{\mathrm{x}}_3	&0	&0	&r^3	&0	&0	&0	&r^2\tilde{r}	&1-(r^3 + r^2\tilde{r})\\
                \bar{\mathrm{x}}_4	&0	&0	&0	&r^3	&r^2\tilde{r}	&0	&0	&1-(r^3 + r^2\tilde{r})\\
                \bar{\mathrm{x}}_5	&0	&0	&0	&0	&r^2	&0	&0	&1-r^2\\
                \bar{\mathrm{x}}_6	&0	&0	&0	&0	&0	&r^3	&r^2\tilde{r}	&1-(r^3+r^2\tilde{r})\\
                \bar{\mathrm{x}}_7	&0	&0	&0	&0	&0	&0	&r^2	&1-r^2\\
                \bar{\mathrm{x}}_8	&0	&0	&0	&0	&0	&0	&0	&1
                \end{bNiceArray} \nonumber\\
                \stackrel{r=0.7}{\Longrightarrow}   \bar{\mathbf{P}}_\mathrm{BC3} & = 
                    \begin{bmatrix}
                        \begin{smallmatrix}
                        0.240 & 0.103 & 0.103 & 0.103 & 0.044 & 0.103 & 0.044 & 0.260 \\
                        \cdot & 0.343 & \cdot & \cdot & 0.147 & \cdot & \cdot & 0.510 \\
                        \cdot & \cdot & 0.343 & \cdot & \cdot & \cdot & 0.147 & 0.510 \\
                        \cdot & \cdot & \cdot & 0.343 & 0.147 & \cdot & \cdot & 0.510 \\
                        \cdot & \cdot & \cdot & \cdot & 0.490 & \cdot & \cdot & 0.510 \\
                        \cdot & \cdot & \cdot & \cdot & \cdot & 0.343 & 0.147 & 0.510 \\
                        \cdot & \cdot & \cdot & \cdot & \cdot & \cdot & 0.490 & 0.510 \\
                        \cdot & \cdot & \cdot & \cdot & \cdot & \cdot & \cdot & 1.000 \\
                        \end{smallmatrix}
                    \end{bmatrix},  \label{eq:P_bar_BC3_r_0.7}
    \end{align}
    where zero values are replaced by the centered dots in the above matrix.
    
    The submatrix ${\mathbf{P}}_\mathrm{BC3}$ is a $(7\times 7)$-sized one-step transition probability matrix between the nonfailed states $\bar{\textrm{x}}_1, \bar{\textrm{x}}_2, \dots, \bar{\textrm{x}}_7\in \bar{\mathcal{X}}_\textrm{BC3}$, which is identical to the one-step transition probability matrix between the nonfailed states $\textrm{x}_{1}, \textrm{x}_{2}, \textrm{x}_{3}, \textrm{x}_{5}, \textrm{x}_{6}, \textrm{x}_{9}, \textrm{x}_{11}\in\mathcal{X}$. The column vector $\mathbf{e}$ is a seven-dimensional vector of ones and $\mathbf{0}$ is a seven-dimensional row vector of zeros. Now, we know from Eq.~\eqref{eq:pmf} in Section~\ref{sec:modeling_analysis_lifetime} that an absorbing DTMC defined by such a transition probability matrix $\bar{\mathbf{P}}_\mathrm{BC3}$ has the following probability mass function for its time to absorption: 
    \begin{align*}
        \mathbb{P}\left\{M=m\right\}=\pmb{\upalpha}_d {\mathbf{P}}_\mathrm{BC3}^{m-1}\left(\mathbf{e}-{\mathbf{P}}_\mathrm{BC3}\mathbf{e}\right)\quad\mathrm{for}\ m=1,2,\dots,
    \end{align*}
    where $M$ is a random variable denoting the number of transitions to absorption and $\pmb{\upalpha}_d=[1,0,\dots,0]$ is a seven-dimensional initial probability (row) vector that corresponds to seven transient states in $\bar{\mathcal{X}}_\textrm{BC3}$. As such, the SNTF of a circular $2$-out-of-$4$: G balanced system follows the following discrete phase-type distribution
    \begin{align}
        M\sim PH_d\left(\pmb{\upalpha}_d=[1,0,...,0], \mathbf{T}_d = {\mathbf{P}}_\mathrm{BC3}\right),
    \end{align}    
    where $\mathbf{P}_\mathrm{BC3}$ is obtained as included in Eq.~\eqref{eq:P_bar_BC3_r_0.7} and therefore the mean SNTF can be calculated by $\mathbb{E}[M]=\sum_{m=0}^\infty m\mathbb{P}\left\{M=m\right\}=\pmb{\upalpha}_d\left(\mathbf{I}-\mathbf{P}_\mathrm{BC3}\right)^{-1}\mathbf{e}$ as explained in Eq~\eqref{eq:msntf}. Thus far, we have derived the probability distribution of the SNTF for the target system, representing its discrete-time lifetime. 

    Next, we will introduce the inter-shock time distribution, $Erlang(\alpha=2,\lambda=2)$, to transform the SNTF distribution into a continuous probability distribution for the time to system failure (TTF). Since the Erlang distribution is interpreted as a mixture of exponential distributions, we first express the inter-shock time $Y\sim Erlang(\alpha=2,\lambda=2)$ by the following equivalent phase-type distribution: 
    \begin{align}
        Y\sim PH_c\left(\pmb{\upalpha}_c=[1,0],\mathrm{\textbf{T}}_c=\begin{bmatrix} -2 & 2 \\ 0 & -2 \end{bmatrix}\right),
    \end{align}
    where $\pmb{\upalpha}_c$ is the initial probability (row) vector for the corresponding continuous-time Markov chain (CTMC) and $\mathrm{\textbf{T}}_c$ is the subgenerator matrix consisting of the transition rates between the phases constructing the CTMC. The TTF, denoted as a random variable $Z$, is expressed as the random sum of $Y_i$ values: $Z=Y_1+Y_2+\cdots+Y_{M}$. According to Proposition~\ref{prop:cp}, this results in another phase-type distribution having new parameters $\pmb{\upalpha}_Z$ and $\mathbf{T}_Z$, as follows:
    % \begin{align}
    %     Z \sim PH_c \left( \pmb{\upalpha}_Z = [1,0,\dots,0]\otimes [1,0],\ \mathbf{T}_Z = \mathbf{I} \otimes \begin{bmatrix} -2 & 2 \\ 0 & -2 \end{bmatrix} + {\mathbf{P}}_\mathrm{BC3} \otimes \left(-\begin{bmatrix} -2 & 2 \\ 0 & -2 \end{bmatrix} \mathbf{e} [1,0]\right) \right), \label{eq:new_CTPH}
    % \end{align}

    \begin{align}
        Z \sim PH_c \Bigg( \pmb{\upalpha}_Z = [1,0,\dots,0]\otimes [1,0],\  
        &\mathbf{T}_Z =
        \mathbf{I} \otimes 
        \begin{bmatrix} -2 & 2 \\ 0 & -2 \end{bmatrix} \nonumber\\
        &\quad+ {\mathbf{P}}_\mathrm{BC3} \otimes 
        \left( -\begin{bmatrix} -2 & 2 \\ 0 & -2 \end{bmatrix} 
        \mathbf{e} [1,0] \right)
        \Bigg), \label{eq:new_CTPH}
    \end{align}

    where $[1,0,\dots,0]$ is a seven-dimensional row vector, $\mathbf{I}$ is a $(7\times 7)$-sized identity matrix, ${\mathbf{P}}_\mathrm{BC3}$ is a $(7\times 7)$-sized matrix as included in Eq.~\eqref{eq:P_bar_BC3_r_0.7}, and $\mathbf{e}$ is a two-dimensional column vector of ones. For the readers' information, we put the explicit numerical values of $\pmb{\upalpha}_Z$ and $\mathbf{T}_Z$ in Eq.~\eqref{eq:new_CTPH} below:
    \begin{gather}
        \pmb{\upalpha}_Z = [1,0,0,0,0,0,0,0,0,0,0,0,0,0]\ \mathrm{and} \nonumber\\ 
        \mathbf{T}_Z =
            \begin{bmatrix}
                {\scriptsize
                \begin{smallmatrix}
                -2.000 & 2.000 & \cdot & \cdot & \cdot & \cdot & \cdot & \cdot & \cdot & \cdot & \cdot & \cdot & \cdot & \cdot \\
                -0.480 & -2.000 & -0.206 & \cdot & -0.206 & \cdot & -0.206 & \cdot & -0.088 & \cdot & -0.206 & \cdot & -0.088 & \cdot \\
                \cdot & \cdot & -2.000 & 2.000 & \cdot & \cdot & \cdot & \cdot & \cdot & \cdot & \cdot & \cdot & \cdot & \cdot \\
                \cdot & \cdot & -0.686 & -2.000 & \cdot & \cdot & \cdot & \cdot & -0.294 & \cdot & \cdot & \cdot & \cdot & \cdot \\
                \cdot & \cdot & \cdot & \cdot & -2.000 & 2.000 & \cdot & \cdot & \cdot & \cdot & \cdot & \cdot & \cdot & \cdot \\
                \cdot & \cdot & \cdot & \cdot & -0.686 & -2.000 & \cdot & \cdot & \cdot & \cdot & \cdot & \cdot & -0.294 & \cdot \\
                \cdot & \cdot & \cdot & \cdot & \cdot & \cdot & -2.000 & 2.000 & \cdot & \cdot & \cdot & \cdot & \cdot & \cdot \\
                \cdot & \cdot & \cdot & \cdot & \cdot & \cdot & -0.686 & -2.000 & -0.294 & \cdot & \cdot & \cdot & \cdot & \cdot \\
                \cdot & \cdot & \cdot & \cdot & \cdot & \cdot & \cdot & \cdot & -2.000 & 2.000 & \cdot & \cdot & \cdot & \cdot \\
                \cdot & \cdot & \cdot & \cdot & \cdot & \cdot & \cdot & \cdot & -0.980 & -2.000 & \cdot & \cdot & \cdot & \cdot \\
                \cdot & \cdot & \cdot & \cdot & \cdot & \cdot & \cdot & \cdot & \cdot & \cdot & -2.000 & 2.000 & \cdot & \cdot \\
                \cdot & \cdot & \cdot & \cdot & \cdot & \cdot & \cdot & \cdot & \cdot & \cdot & -0.686 & -2.000 & -0.294 & \cdot \\
                \cdot & \cdot & \cdot & \cdot & \cdot & \cdot & \cdot & \cdot & \cdot & \cdot & \cdot & \cdot & -2.000 & 2.000 \\
                \cdot & \cdot & \cdot & \cdot & \cdot & \cdot & \cdot & \cdot & \cdot & \cdot & \cdot & \cdot & -0.980 & -2.000 \\
                \end{smallmatrix}
                }
            \end{bmatrix}, \nonumber
        \end{gather}
    where zero values are replaced by the centered dots in the above matrix.

    Using the previously explained results in Step 4 of Section~\ref{sec:modeling_analysis_lifetime}, the probability density function and the $p^\mathrm{th}$ moment of the TTF can be obtained as follows:
    \begin{align*}
        f(z)&= \pmb{\alpha}_Z \exp\left({\mathbf{T}_Z z}\right) (-\mathbf{T}_Z\mathbf{e})\quad\mathrm{for}\ z\ge 0\quad\mathrm{and} \\
        \mathbb{E}[Z^p]&= \int_0^\infty z^p f(z) \mathrm{d}z=p ! \pmb{\upalpha}_Z\left(-\mathbf{T}_Z^{-1}\right)^{p}\mathbf{e}\quad\mathrm{for}\ p=1,2,\dots, 
    \end{align*}
    where the shapes of the $f(z)$ and $\mathbb{E}[Z^p]$ for the target system are depicted as the plots in Fig.~\ref{fig:example_pdf_moment}:
    \begin{figure}[H]
	\centering
	\subfloat[][pdf of TTF]
	{
		\centering\resizebox{0.48\textwidth}{!}{\includegraphics{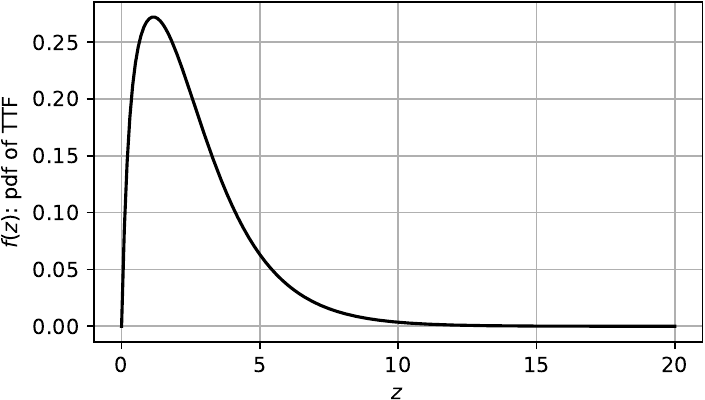}}
		\label{fig:example_pdf}		
	}
	~
	\subfloat[][$p^\mathrm{th}$ moment of TTF ($p=1,2,3,4$)]
	{
		\centering\resizebox{0.48\textwidth}{!}{\includegraphics{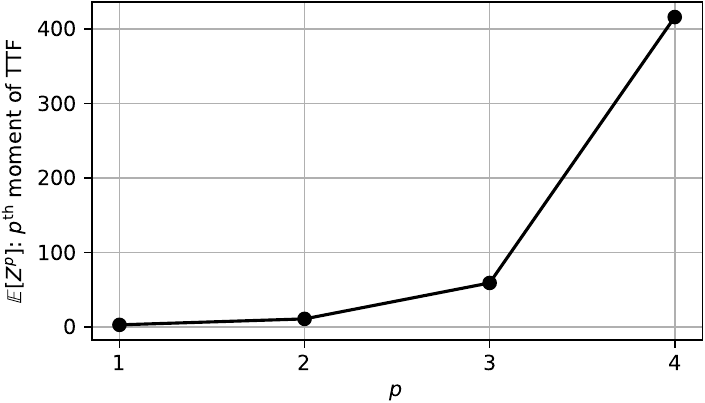}}
		\label{fig:example_moment}		
	}
	\caption{Shapes of pdf and $p^\mathrm{th}$ moment of TTF for the target system with its system parameters $k=2,n=4,r=0.7$, BC3, and $Erlang(\alpha=2,\lambda=2)$}
	\label{fig:example_pdf_moment}
    \end{figure}

    % This section has presented a detailed numerical example illustrating the proposed method explained in Section~\ref{sec:modeling_analysis_lifetime}. The next section expands on this by offering extensive numerical studies on the sensitivities of the system lifetimes from both discrete- and continuous-time perspectives.

    \section{Numerical Results} \label{sec:num_result} 
    This section presents the results of numerical experiments aimed at analyzing system lifetimes. While Section~\ref{sec:descriptive_case} presented a comprehensive analysis of a specific system configuration, this section broadens the investigation to encompass a wider range of system settings. The focus is on examining the behavior of the SNTF and the TTF, including the distribution characteristics and statistical moments. Furthermore, sensitivity analyses are conducted to gain deeper insights into how critical parameters affect system reliability. These analyses aim to reveal underlying patterns and insights, enhancing comprehension of system lifetime dynamics and contributing to improved reliability modeling.    
    
    As discussed in the previous sections, the uncertainties in SNTF and TTF of the CknGB system are characterized using discrete and continuous phase-type distributions, respectively. As the number of units $n$ in the system increases, the size of the transition probability matrix associated with these phase-type distributions grows exponentially. This rapid expansion makes matrix operations—such as matrix exponentiation, inversion, and power computations in Eq.~\eqref{eq:pmf} and Eq.~\eqref{eq:pdf}—computationally intractable.

    To mitigate these challenges, we employ the computationally efficient technique introduced in Section~\ref{sec:modeling_analysis_lifetime}, supplemented by CUDA libraries \cite{CUDA}, enabling high-performance GPU computing for matrix operations. Empirical results indicate that computations remain fully tractable for systems with up to $n = 10$. However, for $n = 12$ or $n = 14$, certain configurations present significantly increased computational burden. Beyond $n = 16$, constraints on computing resources make it impractical to analyze most configurations. Consequently, the numerical results presented in this section are restricted to systems with $n \leq 14$. For the readers' reference, all computational codes are implemented in Julia programming language~\cite{Bezanson2017} and executed on an Intel(R) Core(TM) i7-14700 2.10 GHz processor with 32GB of RAM and an NVIDIA GeForce RTX 4070 GPU.

    \subsection{Analyses of SNTFs} \label{subsec:analysis_SNTF}
    We begin by analyzing the discrete-time system lifetime, or SNTF. Section~\ref{subsubsec:shapes_SNTF} examines the distributional shapes of SNTF under different system parameters. Section~\ref{subsubsec:sensitivity_SNTF} investigates the influence of balance conditions on the mean SNTF, or MSNTF. Section~\ref{subsubsec:mean_SNTF} extends the sensitivity analysis to the system size, exploring how the MSNTF varies across different $(n,k)$ pairs.
    
    \subsubsection{Shapes of distributions} \label{subsubsec:shapes_SNTF}
    As discussed in Section~\ref{sec:modeling_analysis_lifetime}, the pmf of the SNTF depends on the system parameters $n$, $k$, $r$, and the balance condition BC. Among these, $n$, $k$, and BC can be considered intrinsic system parameters, typically determined by the physical characteristics and user preferences. For example, in the UAV system described in Section~\ref{sec:sys_desc}, the total number of units ($n$), the minimum required number of operating units ($k$), and the balance condition (BC) are specified based on physical design choices and operational requirements. In contrast, unit reliability ($r$) can be viewed as a modeling parameter, which may be estimated using historical failure data associated with external shocks. Given this distinction, we first analyze how the pmf of the SNTF varies with changes in $r$, while keeping $n$, $k$, and BC fixed.

    \begin{figure}[H]
    	\centering
    	\subfloat[][$k=4$, $n=12$, BC3]
    	{
    		\centering\resizebox{0.31\textwidth}{!}{\includegraphics{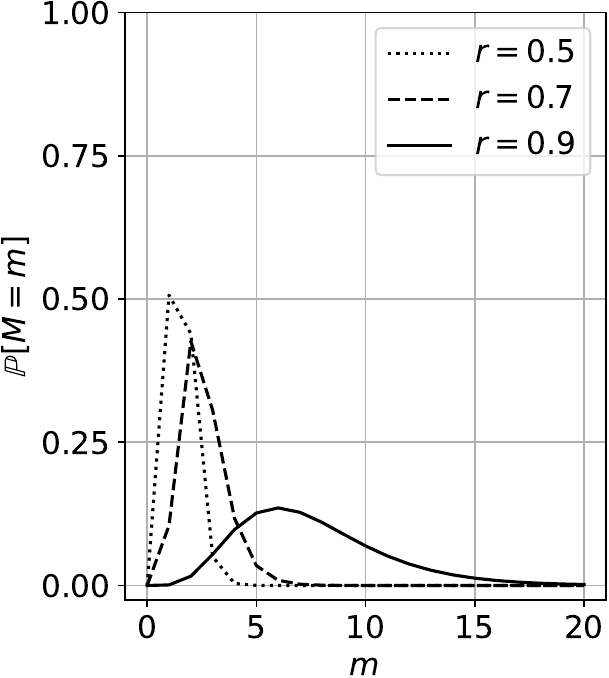}}
    		\label{subfig:pmf_of_SNTF_a}		
    	}
    	~
    	\subfloat[][$k=6$, $n=12$, BC3]
    	{
    		\centering\resizebox{0.31\textwidth}{!}{\includegraphics{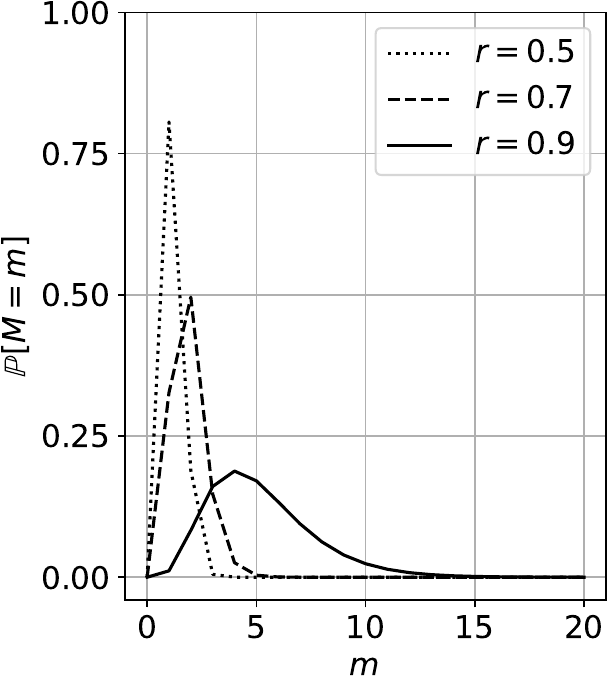}}
    		\label{subfig:pmf_of_SNTF_b}	
    	}
    	~
    	\subfloat[][$k=8$, $n=12$, BC3]
    	{
    		\centering\resizebox{0.31\textwidth}{!}{\includegraphics{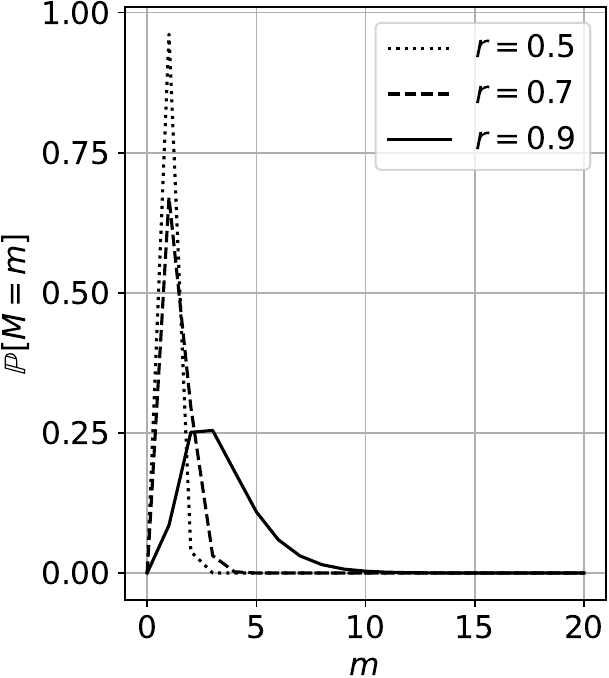}}
    		\label{subfig:pmf_of_SNTF_c}
    	}
    	\caption{Shapes of the pmfs of the SNTFs for circular $(4,6,8)$-out-of-$12$: G balanced systems where $r=0.5,0.7,0.9$ and BC3 is applied as a balance condition}
    	\label{fig:pmf_of_SNTF}
    \end{figure}

    Fig.~\ref{fig:pmf_of_SNTF} illustrates the shapes of the SNTF distributions for $k$-out-of-12 systems under BC3, with $k=4,6,12$. Each subplot shows the pmf curves for three different reliability values, $r=0.5,0.7,0.9$. A common trend observed across all three subplots is that as unit reliability $r$ increases, the probability mass distribution shifts to the right. This indicates that as each unit has a higher probability of surviving a shock, the system as a whole is more likely to remain operational for a longer duration. In other words, the balance condition remains satisfied for an extended period. This result aligns with the intuitive expectation that, given the same number of external shocks, a system with higher unit reliability has a greater chance of surviving.

    Comparing across the subplots, we observe that when $k$ is small, the pmf of the SNTF $M$ is more widely dispersed in Fig.~\ref{subfig:pmf_of_SNTF_a}. As $k$ increases, the probability mass becomes more concentrated as in Figs.~\ref{subfig:pmf_of_SNTF_b}~and~\ref{subfig:pmf_of_SNTF_c}. For the same $r$ value, increasing $k$ results in a lower mode of $M$ with a more pronounced peak in its probability. This phenomenon arises because a higher required number of operating units $k$ reduces the number of successful operational scenarios, making system failure more likely to occur at an earlier stage. Consequently, system failures tend to occur after fewer external shocks and with greater certainty, leading to a more concentrated distribution. On the other hand, for smaller $k$, the greater variety of successful operational scenarios results in a more widely dispersed probability distribution.
    
    \subsubsection{Mean SNTF values under different balance conditions} \label{subsubsec:sensitivity_SNTF}

    Fig.~\ref{fig:mean_SNTF} presents the MSNTF values for circular $k$-out-of-$12$ systems under different balance conditions. Figs.~\ref{fig:mean_SNTF_a}-\ref{fig:mean_SNTF_c} display the results for $k=4,6,8$, respectively, with the x-axis representing unit reliability $r$ and the y-axis denoting the MSNTF $\mathbb{E}[M]$. Each plot includes three curves corresponding to the balance conditions BC1, BC2, and BC3, marked with cross-shaped, triangular, and circular markers respectively. 
    \begin{figure}[t!]
        \centering
    	\subfloat[][$k=4$, $n=12$]
    	{
    		\centering\resizebox{0.31\textwidth}{!}{\includegraphics{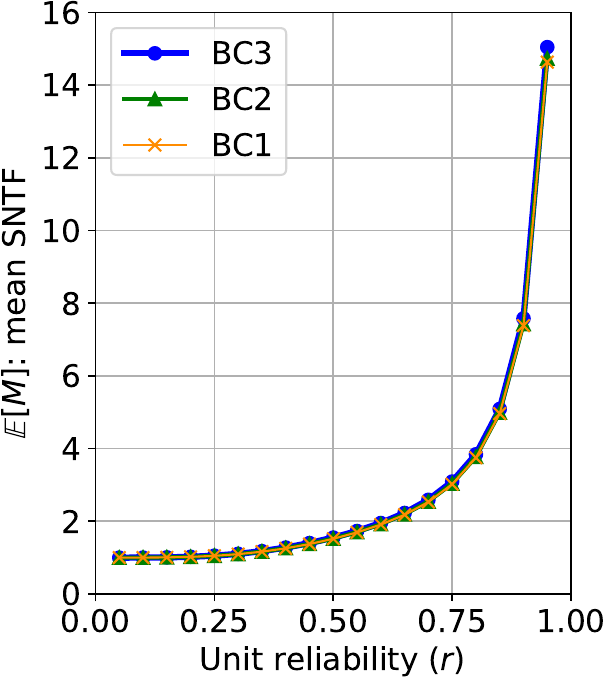}}
    		\label{fig:mean_SNTF_a}		
    	}
    	~
    	\subfloat[][$k=6$, $n=12$]
    	{
    		\centering\resizebox{0.31\textwidth}{!}{\includegraphics{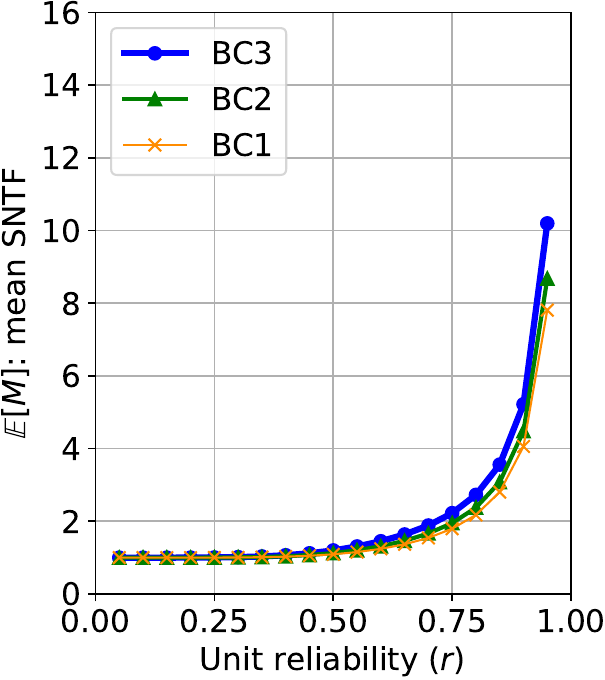}}
    		\label{fig:mean_SNTF_b}		
    	}
    	~
    	\subfloat[][$k=8$, $n=12$]
    	{
    		\centering\resizebox{0.31\textwidth}{!}{\includegraphics{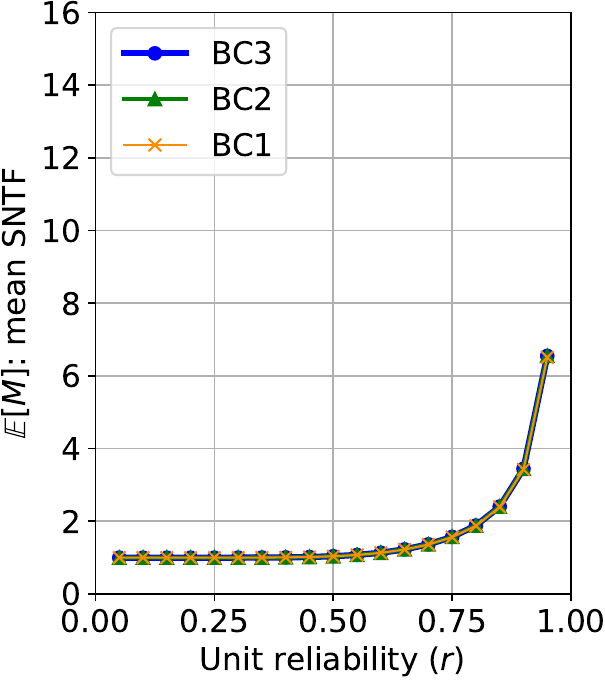}}
    		\label{fig:mean_SNTF_c}		
    	}
        \caption{Comparison between the MSNTF values under different balance conditions for circular $(4,6,8)$-out-of-$12$: G balanced systems}
        \label{fig:mean_SNTF}
    \end{figure} 

    First, let us examine how the MSNTF values change under a fixed balance condition. An anticipatory, but important observation from all the curves in Figs.~\ref{fig:mean_SNTF_a}-\ref{fig:mean_SNTF_c} is that, with $k$ and $n$ fixed, the MSNTF increases exponentially as $r$ increases. This trend aligns with our previous observations in Fig.~\ref{fig:pmf_of_SNTF}, where an increase in $r$ rapidly shifts the SNTF distribution to the right, leading to a larger mean value. Meanwhile, for fixed $n$ and $r$, the MSNTF generally decreases as $k$ increases. This is because a larger $k$ reduces the likelihood of the system satisfying the operational condition, thereby making it more susceptible to sequential shocks. 

    We then examine the impact of different balance conditions on the system's MSNTF. The most striking observation is that, when all other system parameters remain the same, applying different balance conditions definitely results in a difference in MSNTF, as shown in Fig.~\ref{fig:mean_SNTF_b}. This phenomenon occurs because, as explained through Section~\ref{subsec:ro} and Section~\ref{subsec:sys_rel}, the system reliability of the CknGB system is determined by the number of minimum tie-sets, denoted by $|\mathcal{T}|$. Moreover, although the differences are not dramatically pronounced, there is a consistent tendency where MSNTF follows the order BC3 $\ge$ BC2 $\ge$ BC1. This observation aligns precisely with the findings of Cho~\textit{et al.}~\cite{cho2023}, which demonstrated that for the same CknGB system, the number of minimum tie-sets generated by the applied balance condition was highest for BC3, followed by BC2, and then BC1.

    \subsubsection{Overall sensitivity analysis on the MSTNF} \label{subsubsec:mean_SNTF}
    Next, we conduct an overall sensitivity analysis of the MSNTF values. Fig.~\ref{fig:reliability_comparison_type_3} visualizes the MSNTF values for circular $k$-out-of-$n$: G systems across varying $k$ and $n$, under different balance conditions and unit reliability levels. Each subplot corresponds to a specific balance condition (BC1, BC2, or BC3) and reliability level ($r=0.5, 0.7, 0.9$), as labeled. Within each surface plot, the x- and y-axes represent $k$ and $n$, respectively, with the range restricted to $2 \leq k \leq n-1$. The z-axis indicates the MSNTF $\mathbb{E}[M]$. These surface plots provide a visual overview of how MSNTF varies across plausible combinations of $(n,k)$ pairs.
    \begin{figure}[H]
    	\centering
    	\subfloat[][BC1, $r=0.5$]
    	{
    		\centering\resizebox{0.31\textwidth}{!}{\includegraphics{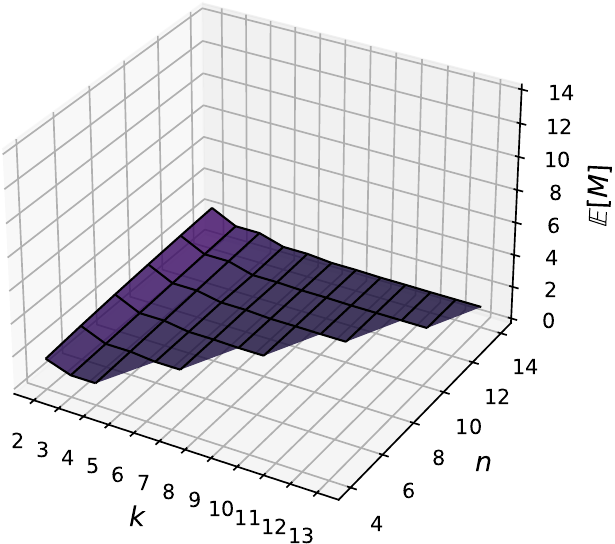}}
    	}
    	~
    	\subfloat[][BC1, $r=0.7$]
    	{
    		\centering\resizebox{0.31\textwidth}{!}{\includegraphics{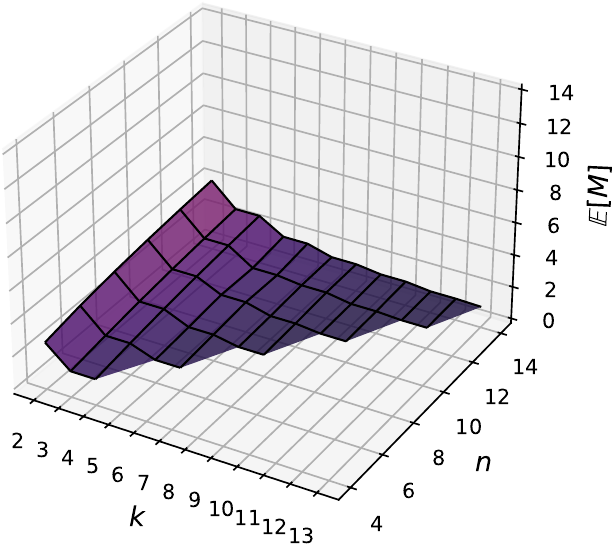}}
    	}
    	~
    	\subfloat[][BC1, $r=0.9$]
    	{
    		\centering\resizebox{0.31\textwidth}{!}{\includegraphics{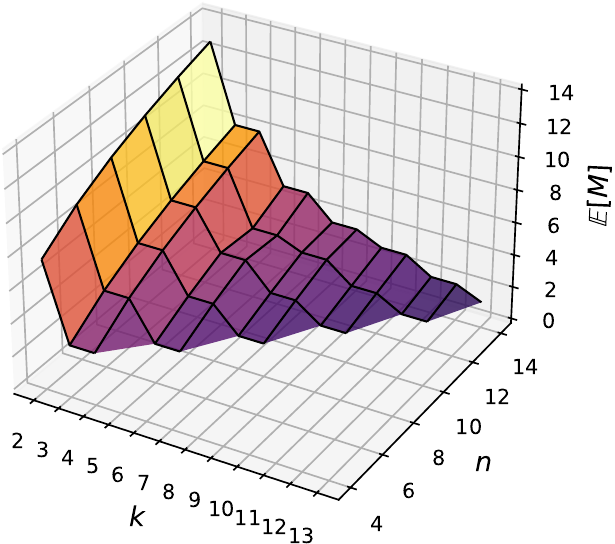}}
    	}
    	
    	\subfloat[][BC2, $r=0.5$]
    	{
    		\centering\resizebox{0.31\textwidth}{!}{\includegraphics{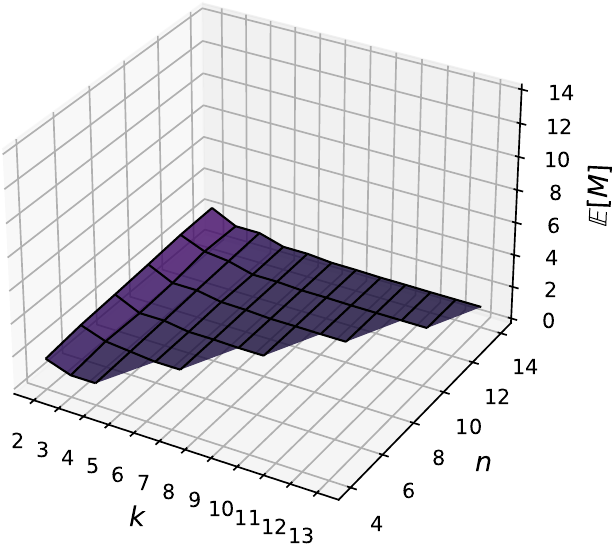}}
    	}
    	~
    	\subfloat[][BC2, $r=0.7$]
    	{
    		\centering\resizebox{0.31\textwidth}{!}{\includegraphics{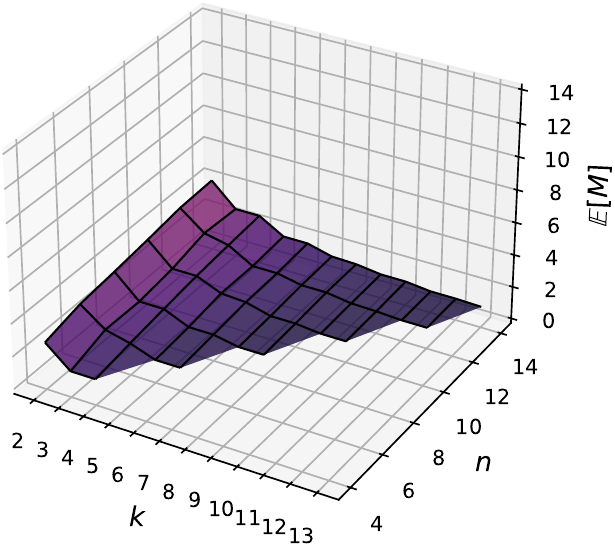}}
    	}
    	~
    	\subfloat[][BC2, $r=0.9$]
    	{
    		\centering\resizebox{0.31\textwidth}{!}{\includegraphics{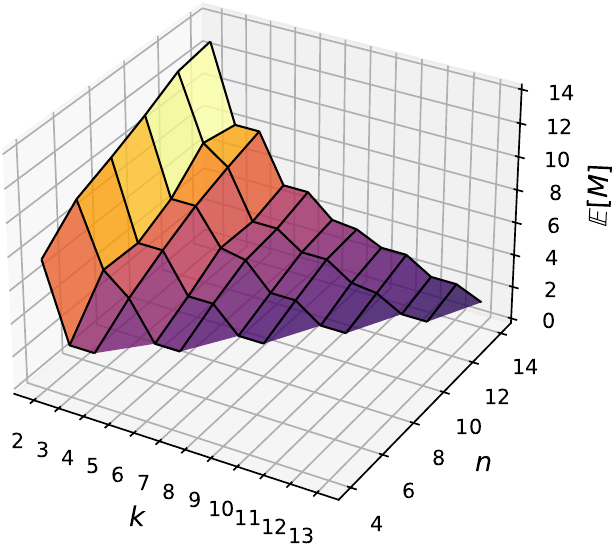}}
    	}
    	
    	\subfloat[][BC3, $r=0.5$]
    	{
    		\centering\resizebox{0.31\textwidth}{!}{\includegraphics{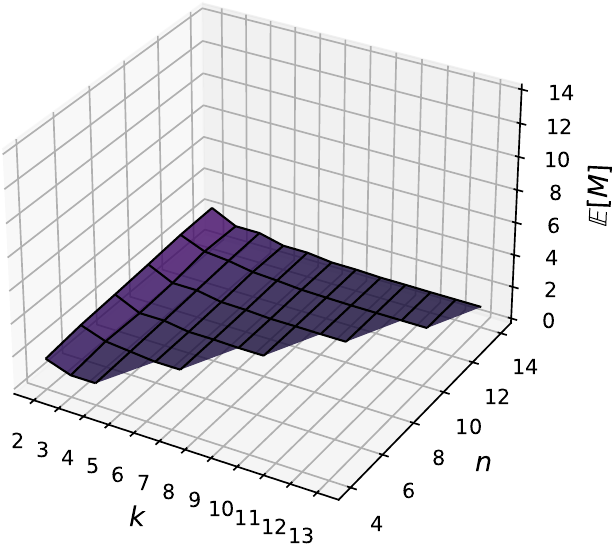}}
    	}
    	~
    	\subfloat[][BC3, $r=0.7$]
    	{
    		\centering\resizebox{0.31\textwidth}{!}{\includegraphics{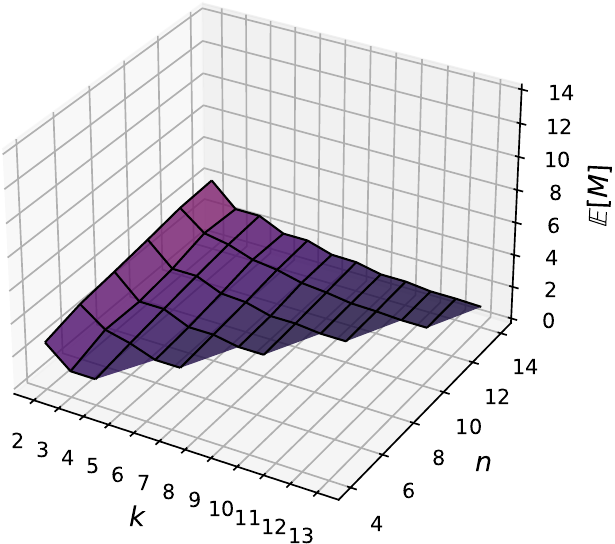}}
    	}
    	~
    	\subfloat[][BC3, $r=0.9$]
    	{
    		\centering\resizebox{0.31\textwidth}{!}{\includegraphics{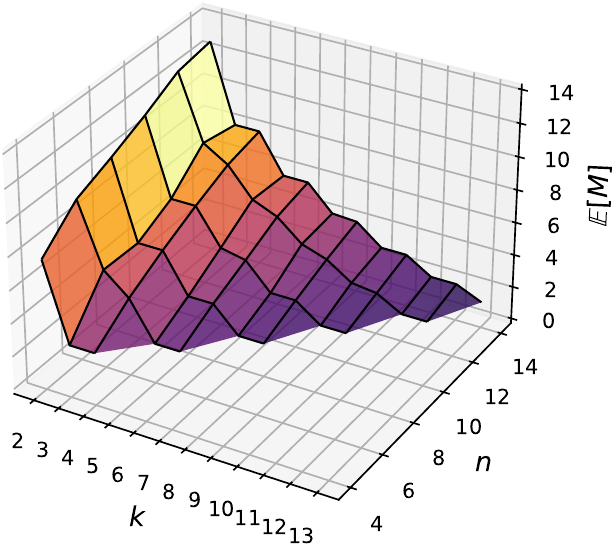}}
    	}
    	\caption{Comparison between the MSNTF values for the systems with varying $k$ and $n$ under $r=0.5,0.7,0.9$ and each balance condition}
    	\label{fig:reliability_comparison_type_3}
    \end{figure}

    As expected, the observed trends align with our earlier findings. Within a specific subplot (i.e., for a fixed balance condition and unit reliability), the MSNTF generally increases as $k$ decreases under a fixed $n$. Examining the subplots across a row (i.e., for a fixed balance condition and $(n,k)$ pair), the MSNTF also increases as $r$ becomes higher. When comparing across balance conditions, the MSNTF values for BC1 and BC2 are generally similar, while BC3 consistently produces higher values due to its less restrictive failure criteria. This behavior is also linked to differences in the number of minimum tie-sets under the respective balance conditions.

    \subsection{Analyses of TTFs} \label{subsec:analysis_TTF}
    We now shift our focus to the continuous-time system lifetime, or TTF. As discussed in Section~\ref{sec:modeling_analysis_lifetime}, if the inter-arrival time between shocks, denoted by a random variable $Y$, follows a phase-type distribution, the derivation results of SNTF can be further extended to analyze TTF, which can also be characterized as a phase-type random variable. Before analyzing the TTF itself, we first describe the underlying inter-shock time distributions.

    We consider three different base inter-shock time distributions with distinct characteristics: Erlang (ER), Exponential (EXP), and Hyperexponential (HE). Although these distributions differ in their properties, they all belong to the family of phase-type distributions. Table~\ref{tab:base_inter_shock_time_distributions} provides the details of these distributions, including their respective parameters. For comparative purposes regarding distributional variation, all three distributions are configured to have the same mean of $1$ while differing in their squared coefficient of variation (SCV), with values of $0.5, 1$, and $2$. The pdfs of these distributions, denoted by $f(y)$, are depicted in Fig.~\ref{fig:figure_of_inter_arrival_time_distributions}, illustrating their distinct distributional shapes.

    \begin{table}[H]
        \centering
        \caption{Base inter-shock time distributions}
        \resizebox{\textwidth}{!}{%
        \begin{tabular}{|c|cc|c|c|c|}
        \hline
        \multirow{2}{*}{\textbf{label}} &
          \multicolumn{2}{c|}{\textbf{$PH_c$ distribution}} &
          \multirow{2}{*}{\textbf{mean}} &
          \multirow{2}{*}{\textbf{SCV}} &
          \multirow{2}{*}{\textbf{Corresponding common distribution}} \\ \cline{2-3}
          & \multicolumn{1}{c|}{\textbf{initial state vector}} & \textbf{subgenerator matrix} &   &     &                         \\ \hline
        ER & \multicolumn{1}{c|}{$\pmb{\upalpha}_c=[1,0]$}       & $\mathrm{\textbf{T}}_c=\begin{bmatrix} -2 & 2 \\ 0 & -2 \end{bmatrix}$           & 1 & 0.5 & $Erlang(\alpha=2,\lambda=2)$ \\[0.3cm] \hline
        EXP & \multicolumn{1}{c|}{$\pmb{\upalpha}_c=[1]$}       & $\mathrm{\textbf{T}}_c=\begin{bmatrix} -1 \end{bmatrix}$           & 1 & 1   & $Exponential(\lambda=1)$        \\[0.3cm] \hline
        HE &
          \multicolumn{1}{c|}{$\pmb{\upalpha}_c=[1/2,1/2]$} &
          $\mathrm{\textbf{T}}_c=\begin{bmatrix} -2/(2-\sqrt{2}) & 0 \\ 0 & -2/(2+\sqrt{2}) \end{bmatrix}$ &
          1 &
          2 &
          $Hyperexponential\left(\pmb{\upalpha}=\begin{bmatrix}
    1/2 \\
    1/2
\end{bmatrix}, \pmb{\uplambda}=\begin{bmatrix}
    2/(2-\sqrt{2}) \\
    2/(2+\sqrt{2})
\end{bmatrix}\right)$ \\[0.3cm] \hline
        \end{tabular}%
        }
        \label{tab:base_inter_shock_time_distributions}
    \end{table}

    \begin{figure}[H]
        \centering
        \includegraphics[width = \textwidth]{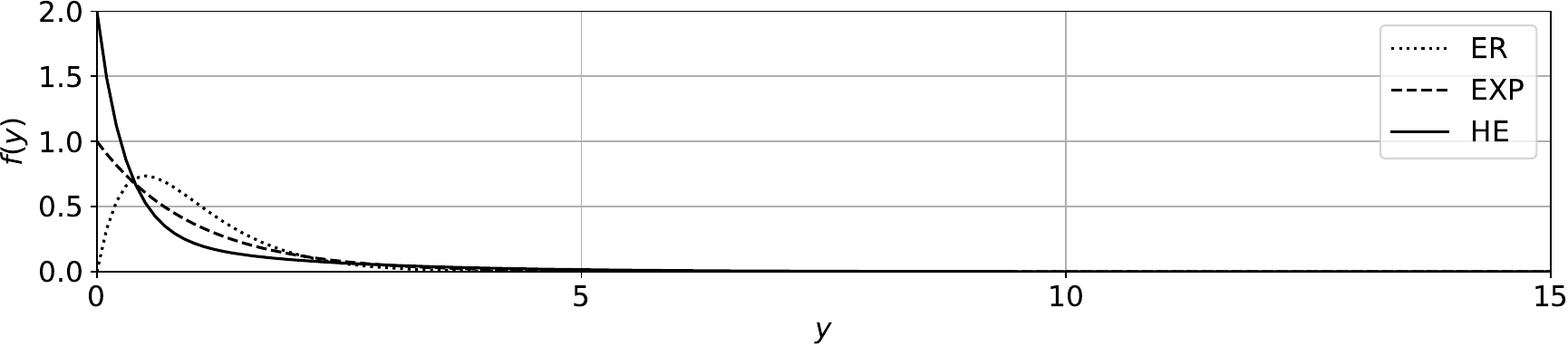}
        \caption{Shapes of the probability density functions of base inter-shock time distributions}
        \label{fig:figure_of_inter_arrival_time_distributions}
    \end{figure} 
    
    \subsubsection{Shapes of distributions under different inter-shock time distributions} \label{subsubsec:shapes_of_TTF}
    The subplots in Fig.~\ref{fig:reliability_comparison_type_4} depict the resulting TTF distributions for $k$-out-of-$12$ systems with $k=4,6,8$ and $r=0.5, 0.7, 0.9$, under the application of BC3 for illustrative purpose. As the underlying base distributions share the same mean, all the TTFs also have the same mean, which is indicated by a purple-colored vertical line with such a value in each subplot. Each row and column of subplots iterates over $r$ and $k$ values, respectively.

    \begin{figure}[H]
    	\centering
    	\subfloat[][$r=0.5$, $4$-out-of-$12$]
    	{
    		\centering\resizebox{0.31\textwidth}{!}{\includegraphics{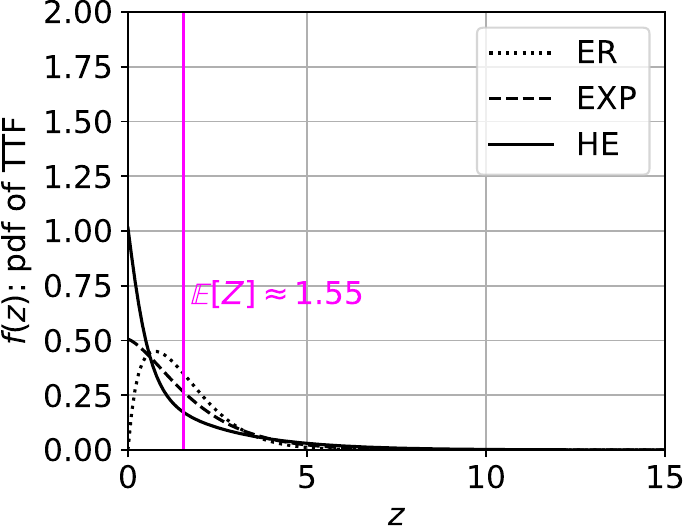}}
    		\label{fig:reliability_comparison_type_4_a}		
    	}
    	~
    	\subfloat[][$r=0.5$, $6$-out-of-$12$]
    	{
    		\centering\resizebox{0.31\textwidth}{!}{\includegraphics{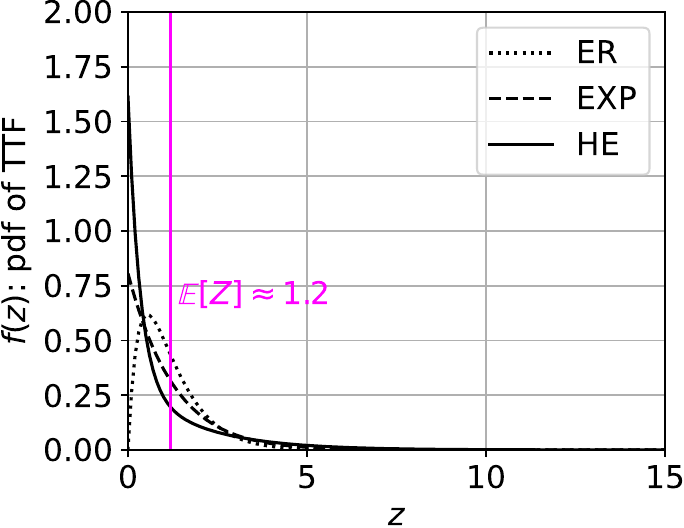}}
    		\label{fig:reliability_comparison_type_4_b}		
    	}
    	~
    	\subfloat[][$r=0.5$, $8$-out-of-$12$]
    	{
    		\centering\resizebox{0.31\textwidth}{!}{\includegraphics{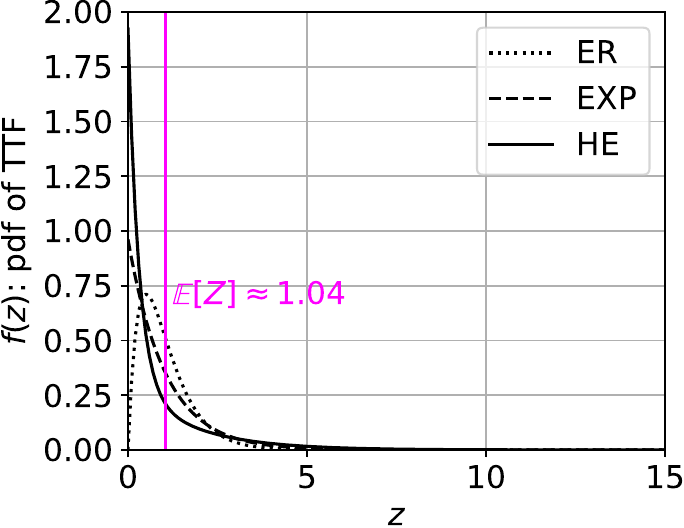}}
    		\label{fig:reliability_comparison_type_4_c}		
    	}
    	
    	\subfloat[][$r=0.7$, $4$-out-of-$12$]
    	{
    		\centering\resizebox{0.31\textwidth}{!}{\includegraphics{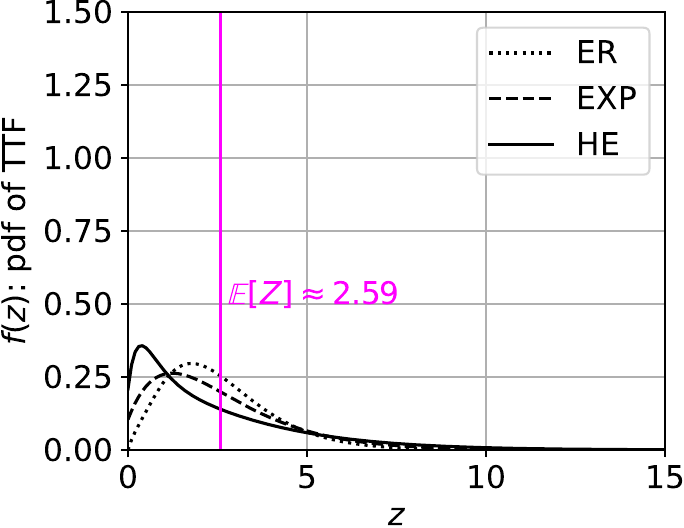}}
    		\label{fig:reliability_comparison_type_4_d}		
    	}
    	~
    	\subfloat[][$r=0.7$, $6$-out-of-$12$]
    	{
    		\centering\resizebox{0.31\textwidth}{!}{\includegraphics{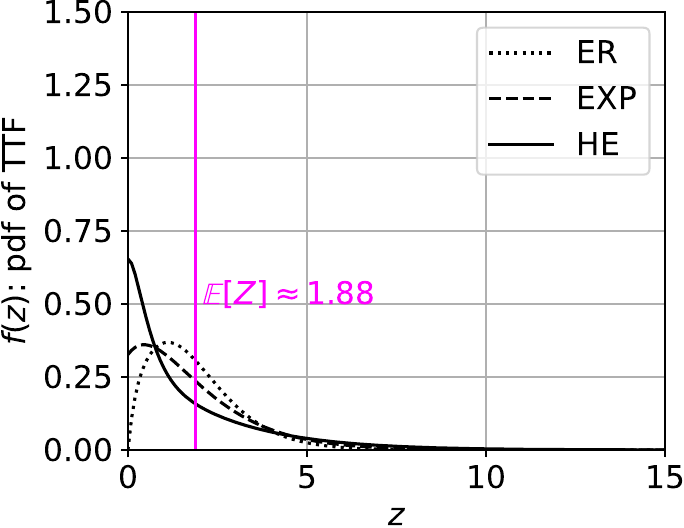}}
    		\label{fig:reliability_comparison_type_4_e}		
    	}
    	~
    	\subfloat[][$r=0.7$, $8$-out-of-$12$]
    	{
    		\centering\resizebox{0.31\textwidth}{!}{\includegraphics{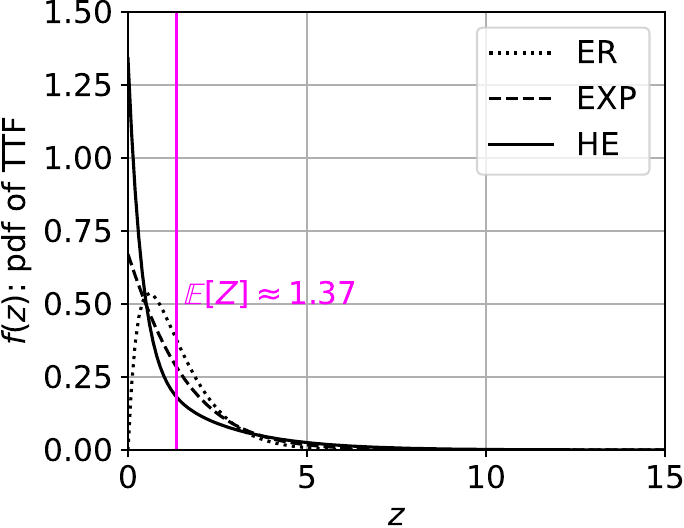}}
    		\label{fig:reliability_comparison_type_4_f}		
    	}
    	
    	\subfloat[][$r=0.9$, $4$-out-of-$12$]
    	{
    		\centering\resizebox{0.31\textwidth}{!}{\includegraphics{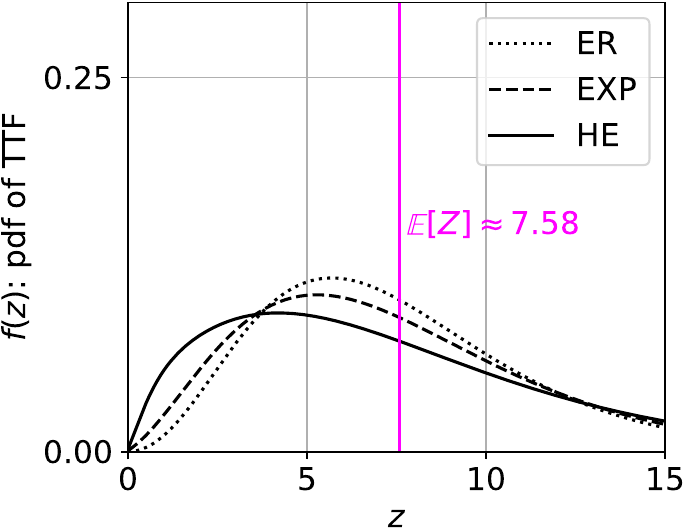}}
    		\label{fig:reliability_comparison_type_4_g}		
    	}
    	~
    	\subfloat[][$r=0.9$, $6$-out-of-$12$]
    	{
    		\centering\resizebox{0.31\textwidth}{!}{\includegraphics{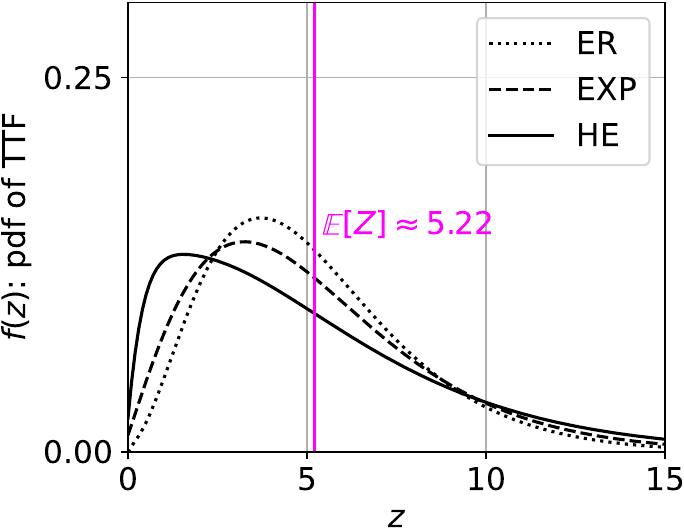}}
    		\label{fig:reliability_comparison_type_4_h}		
    	}
    	~
    	\subfloat[][$r=0.9$, $8$-out-of-$12$]
    	{
    		\centering\resizebox{0.31\textwidth}{!}{\includegraphics{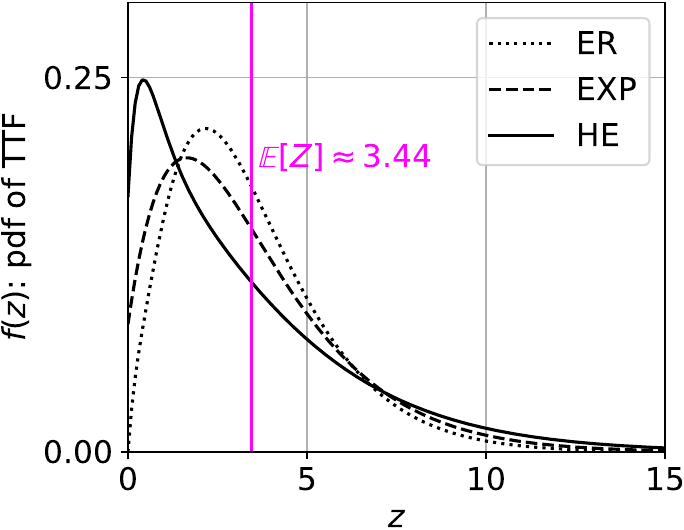}}
    		\label{fig:reliability_comparison_type_4_i}		
    	}
	
	   \caption{Comparison between the probability density functions of TTFs for circular $(4,6,8)$-out-of-$12$: balanced systems under BC3 and three different inter-shock time distributions where $r=0.5,0.7,0.9$ }
	   \label{fig:reliability_comparison_type_4}
    \end{figure}
     
    A consistent pattern is observed in the shape of TTF distributions across all system configurations. To be specific, the TTF distributions appear to inherit the density concentration patterns of their respective base distributions. As shown in Fig.~\ref{fig:figure_of_inter_arrival_time_distributions}, the base distributions---ER, EXP, and HE---exhibit distinct density concentration tendencies: ER peaks at the largest $z$ value, EXP is more evenly spread, and HE is primarily concentrated at smaller $z$ values. This behavior arises from the way in which the TTF is defined as a random sum of inter-shock times, $Z = \sum_{m=1}^M Y_m$, where the number of summands $M$, which is the SNTF. Since each inter-shock time $Y_m$ follows the base inter-shock time distribution and $M$ is consistent across all cases within a given system configuration, the pdf of $Z$ naturally reflects the density characteristics of its underlying base distribution.

    Moreover, the extent of this shift is influenced by the magnitude of system parameters. Higher unit reliability ($r$) leads to a more pronounced shift toward larger $z$ values, as the system remains operational for longer durations. Conversely, smaller $k$ values result in a weaker shift, as the system reaches the failed state more quickly. These phenomena are consistent with those observed in Section~\ref{subsec:analysis_SNTF}, where higher $r$ or smaller $k$ were shown to prolong or shorten system operational durations. This shared tendency between discrete-time (SNTF) and continuous-time (TTF) analyses further highlights the consistency of these reliability measures. 
    
    When it comes to the MTTF, denoted as $\mathbb{E}[Z]$, it increases as $r$ increases or $k$ decreases. One noticeable observation is that the degree of increase in $\mathbb{E}[Z]$ is more sensitive to changes in $r$ than in $k$. In other words, if we examine a specific row of the subplots, for example, the first row with $r=0.5$ in Figs.~\ref{fig:reliability_comparison_type_4_a},~\ref{fig:reliability_comparison_type_4_b}, and~\ref{fig:reliability_comparison_type_4_c}, the MTTF increases from $1.04$ to $1.2$ and then to $1.55$ as $k$ decreases from $8$ to $6$ and then to $4$. However, if we examine a specific column, for example, the first column with $k=4$ in Figs.~\ref{fig:reliability_comparison_type_4_a},~\ref{fig:reliability_comparison_type_4_d}, and~\ref{fig:reliability_comparison_type_4_g}, the MTTF increases to a much greater extent, varying from $1.55$ to $2.59$ and then to $7.58$ as $r$ increases from $0.5$ to $0.7$ and then to $0.9$. This observation suggests that the transition probability matrix $\bar{\mathbf{P}}_{\textrm{BC}}$ exhibits greater sensitivity to change in $r$ compared to the sensitivity of minimum tie-sets to changes in $k$. Thus, an increase in $r$ induces a more pronounced shift in the TTF distributions compared to a decrease in $k$. 

    Comparing different values of $k$ at fixed $r$, the relative increase in the MTTF is greater for higher values of $r$. In particular, for $r=0.5, 0.7, 0.9$, the MTTF at $k=4$ is approximately $1.49$, $1.89$, and $2.20$ times larger, respectively, than that at $k=8$. Similarly, for $k=8,6,4$, the MTTF at $r=0.9$ is approximately $3.31$, $4.35$, and $4.89$ times larger, respectively, than that at $r=0.5$. These results indicate that the effect of increasing $r$ is stronger at smaller $k$ values, and the effect of decreasing $k$ is stronger when $r$ is already high. This suggests that, for extending system operational duration, it is more effective to simultaneously increase $r$ and decrease $k$ rather than applying only one of these changes. 

    \subsubsection{Overall sensitivity analysis on the SCV values} \label{subsubsec:sensitivity_TTF}
    Previously, we discussed the shape of TTF distributions and analyzed their mean values. Now, we examine the squared coefficient of variation (SCV) of TTF, defined as $Var[Z]/(\mathbb{E}[Z])^2$. Similar to the sensitivity analysis in Section~\ref{subsubsec:sensitivity_SNTF}, we investigate how SCV values vary across plausible $(n,k)$ combinations, where $2 \leq k \leq n-1$. Fig.~\ref{fig:sensitivity_analysis_SCV} presents a heatmap of SCV values under a fixed unit reliability $r=0.9$. Each row corresponds to a different base inter-shock time distribution (ER, EXP, and HE), while each column represents a different balance condition (BC1, BC2, and BC3).
    \begin{figure}[H]
    	\centering
    	\subfloat[][ER, BC1]
    	{
    		\centering\resizebox{0.31\textwidth}{!}{\includegraphics{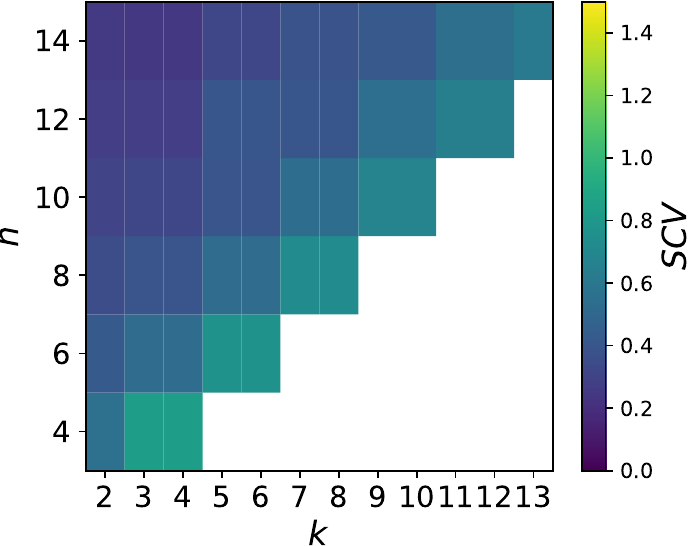}}
    	}
    	~
    	\subfloat[][ER, BC2]
    	{
            \centering\resizebox{0.31\textwidth}{!}{\includegraphics{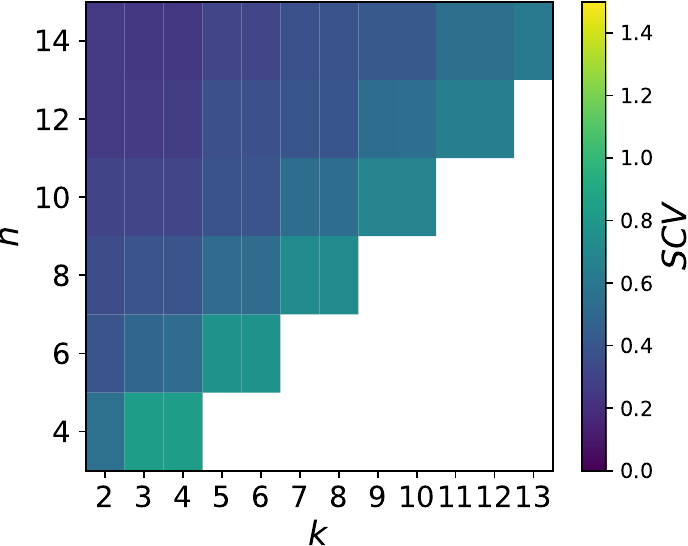}}
    	}
    	~
    	\subfloat[][ER, BC3]
    	{
            \centering\resizebox{0.31\textwidth}{!}{\includegraphics{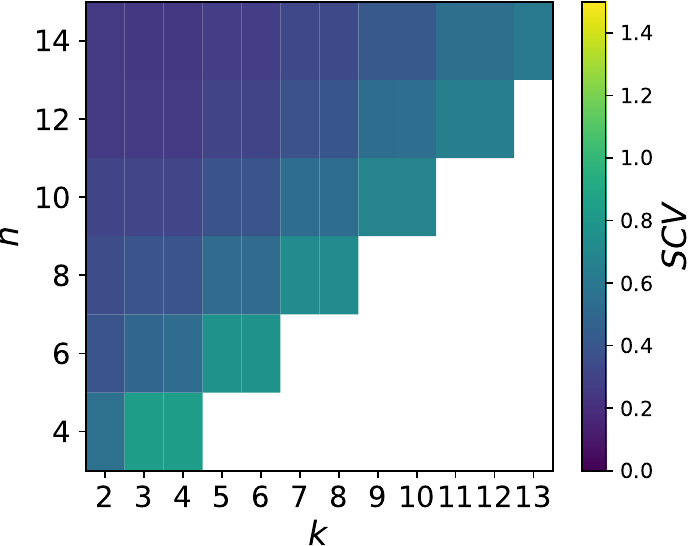}}
    	}
    	
    	\subfloat[][EXP, BC1]
    	{
            \centering\resizebox{0.31\textwidth}{!}{\includegraphics{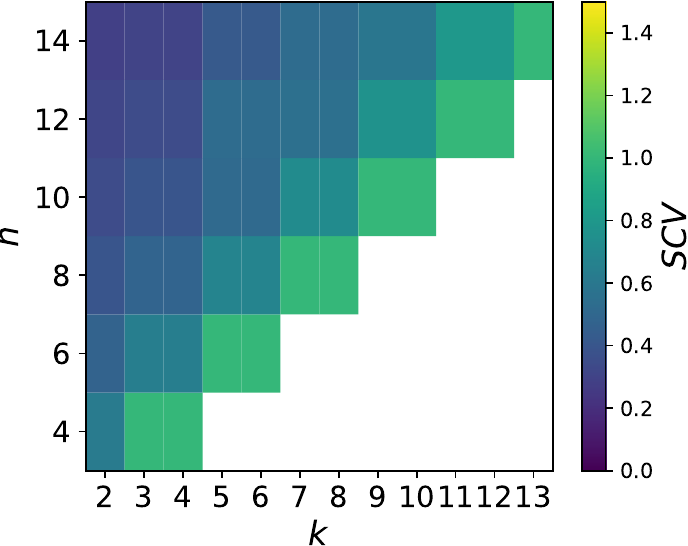}}
    	}
    	~
    	\subfloat[][EXP, BC2]
    	{
            \centering\resizebox{0.31\textwidth}{!}{\includegraphics{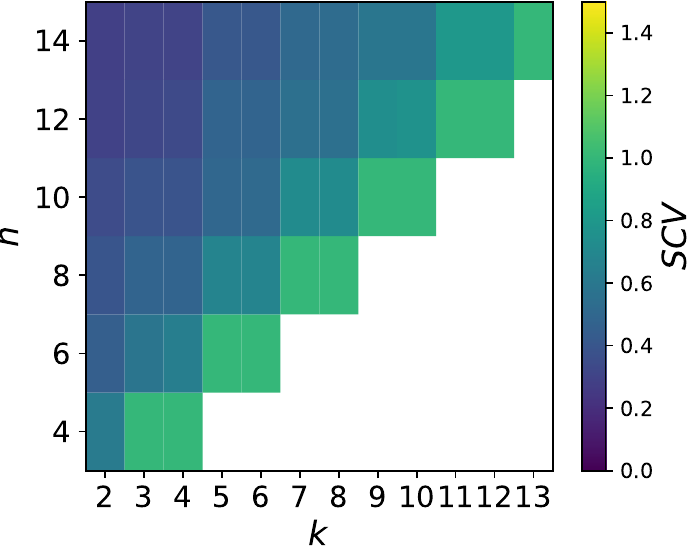}}
    	}
    	~
    	\subfloat[][EXP, BC3]
    	{
            \centering\resizebox{0.31\textwidth}{!}{\includegraphics{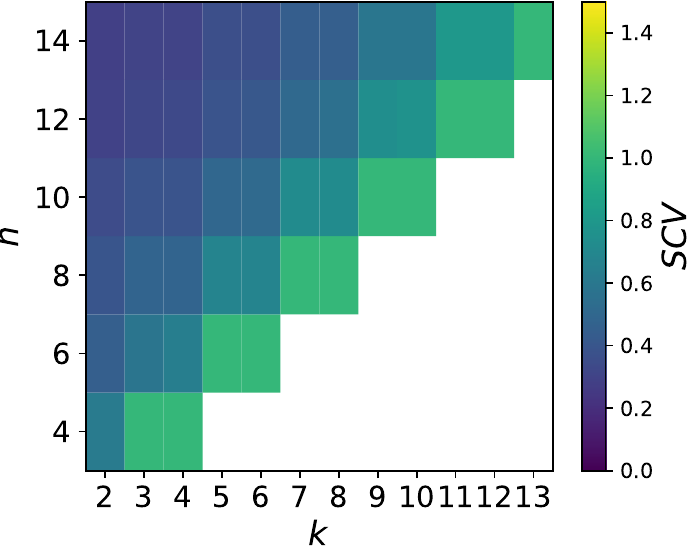}}
    	}
    	
    	\subfloat[][HE, BC1]
    	{
    	   \centering\resizebox{0.31\textwidth}{!}{\includegraphics{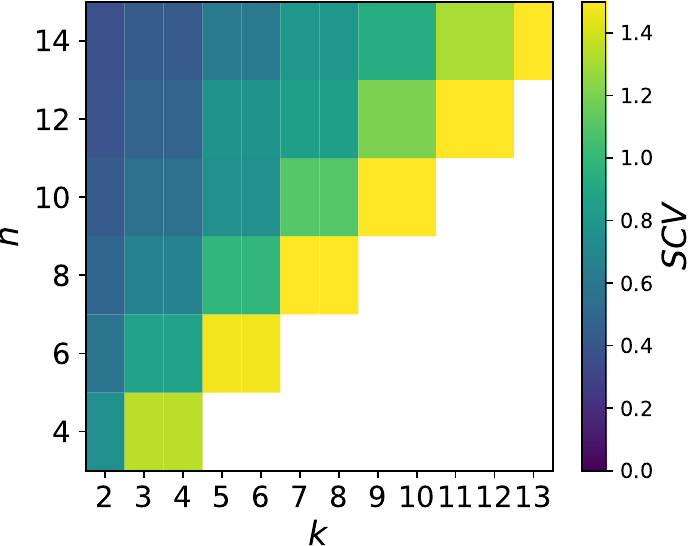}}
    	}
    	~
    	\subfloat[][HE, BC2]
    	{
    		\centering\resizebox{0.31\textwidth}{!}{\includegraphics{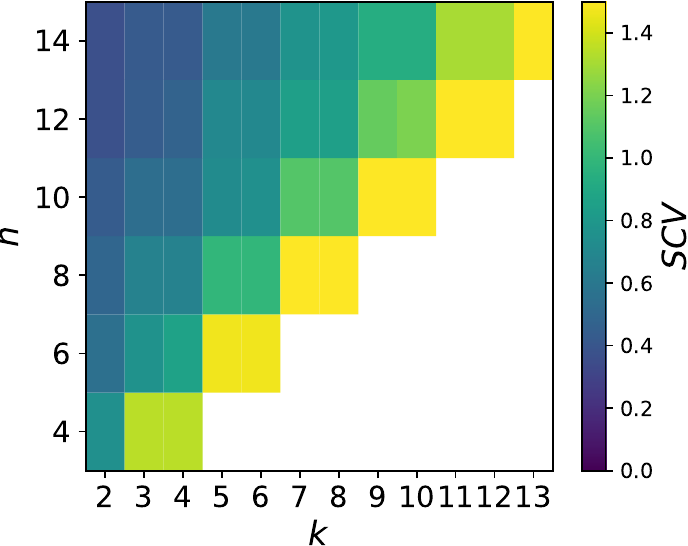}}
    	}
    	~
    	\subfloat[][HE, BC3]
    	{
    		\centering\resizebox{0.31\textwidth}{!}{\includegraphics{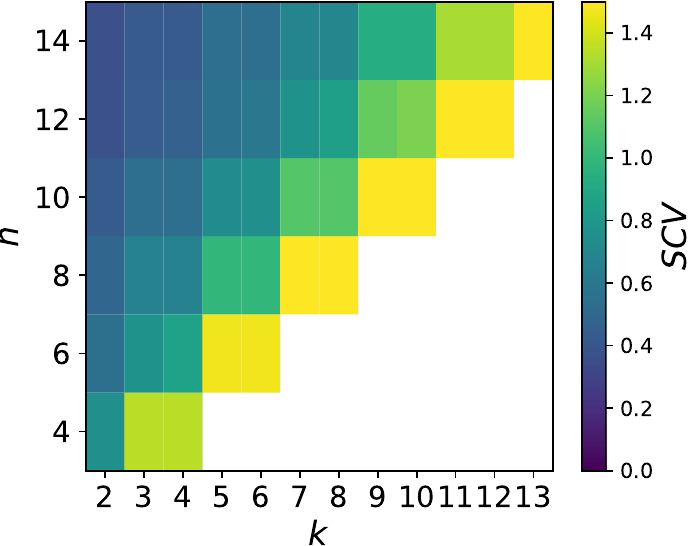}}
    	}
    	
    	\caption{Comparison between the SCV values of TTFs for the systems under $r=0.9$ with varying $k,n$ and for each inter-shock time distribution and balance condition}
    	\label{fig:sensitivity_analysis_SCV}
    \end{figure}

    A consistent trend is observed across all cases. For a fixed $n$, the SCV increases as $k$ increases, and for a fixed $k$, the SCV increases as $n$ decreases. Recalling that the MTTF showed the opposite trend---where it increase as $k$ decreases for fixed $n$, and as $n$ increases for fixed $k$---this result highlights a clear inverse relationship between the SCV and the MTTF. Furthermore, the extent of this effect (i.e., the magnitude of SCV's increment) becomes more pronounced in the increasing order of ER-EXP-HE, as observed by comparing the rows of Fig.~\ref{fig:sensitivity_analysis_SCV}. In other words, this inverse relationship is more evident when the MTTF value is higher.

    \section{Conclusion} \label{sec:conclusion}
    This study focused on the lifetime analysis of circular $k$-out-of-$n$: G balanced systems operating in a shock environment. Unlike traditional $k$-out-of-$n$ systems, these systems require not only a minimum number of operating units but also a predefined balance condition to maintain functionality. By leveraging a two-step finite Markov chain imbedding approach, we derived closed-form expressions for the system’s lifetime distributions in both discrete- and continuous-time settings. The shock-induced failure process was examined under different balance conditions, including symmetry-based (BC1), proportionality-based (BC2), and center of gravity-based (BC3) criteria. One of the key contributions of this work is the computational efficiency improvement in handling large state spaces through balance condition-based state space consolidation and direct calculation of multi-step transition probabilities. This allows for more scalable reliability evaluations of complex balanced systems. Through extensive numerical analyses, we demonstrated that different balance conditions influence the system’s lifetime, with stricter balance constraints leading to higher system vulnerability. Moreover, our results highlight the interplay between key system parameters, such as the number of units ($n$), required minimum operating units ($k$), unit reliability ($r$), and inter-shock time distribution, in determining the overall reliability of the system. The findings from this research can be particularly beneficial in the design and maintenance planning of balance- and safety-critical engineering systems, such as UAVs, UAMs, and aerospace propulsion systems, where geometric balance conditions play a crucial role in operational reliability.
    
    There are several potential extensions to this research that can further enhance the reliability research study of circular $k$-out-of-$n$: G balanced systems. One promising direction is the incorporation of internal degradation mechanisms, where individual units deteriorate over time, independent of external shocks. This would allow for a more comprehensive understanding of system failure progression under real-world operating conditions. Additionally, one could introduce more advanced shock models that account for directional or heterogeneous shock impacts, which better reflects real-world operational environments where different components may experience varying levels of stress. Another area of future study involves cost and reward modeling, where optimal maintenance policies can be developed by integrating cost structures and reward functions to improve system longevity while minimizing operational expenses. Finally, further investigation on computational efficiency are essential for scaling the proposed approach to larger and more complex systems. Investigating numerical methods such as advanced state-space reduction strategies could significantly enhance the tractability of the reliability evaluation process.

% \section*{CRediT authorship contribution statement}
% \textbf{Seung Min Baik:} Conceptualization, Formal analysis, Investigation, Methodology, Software, Validation, Writing-original draft, Writing-review and editing. \textbf{Yongkyu Cho:} Conceptualization, Data curation, Formal analysis, Investigation, Methodology, Software, Supervision, Validation, Visualization, Writing-original draft, Writing-review and editing.

% \section*{Declaration of competing interest}
% The authors declare that they have no known competing financial interests or personal relationships that could have appeared to influence the work reported in this paper.

\section*{Acknowledgment} \label{sec:ack}
This work was supported by Kyonggi University Research Grant 2024.

\bibliographystyle{elsarticle-num} 
\bibliography{manuscript}

\end{document}